\documentclass[journal,a4paper, onecolumn]{IEEEtran}
\usepackage[T1]{fontenc}
\usepackage{amsmath,amsthm,amsfonts,amssymb,amscd}
\usepackage[colorlinks,
            linkcolor=blue,       
            anchorcolor=blue,  
            citecolor=blue,        
            ]{hyperref}
\usepackage{graphicx}
\usepackage{zref-clever}
\zcsetup{cap}
\usepackage[caption=false,font=footnotesize]{subfig}
\usepackage{booktabs}                  
\usepackage{lineno} %  For line numbers
\usepackage{bbm}
\usepackage{enumitem}
\usepackage{braket}
\usepackage{tikz}
\usepackage{graphicx}
\usepackage{float}
\usepackage{lscape} % Rotates content (table), but KEEPS PDF page portrait
\usepackage{xltabular} % Combines longtable and tabularx features

\usepackage{array}
\usepackage{nicematrix}
\usepackage{tabularx}
\usepackage{makecell}
\usepackage{adjustbox}

\usepackage{arydshln}
\definecolor{myblue}{RGB}{46,86,182}
\definecolor{mygreen}{RGB}{93,141,55}
\definecolor{myorange}{RGB}{230,165,110}
\newtheorem{thm}{Theorem}[section]  
  
\newtheorem{cor}[thm]{Corollary}    
\newtheorem{lemma}[thm]{Lemma}        
\newtheorem{ex}[thm]{Example}       
\newtheorem{rk}[thm]{Remark} 
\newtheorem{Def}[thm]{Definition}

\AddToHook{env/thm/begin}{\zcsetup{countertype={thm=theorem}}}
\AddToHook{env/prop/begin}{\zcsetup{countertype={thm=proposition}}}
\AddToHook{env/cor/begin}{\zcsetup{countertype={thm=corollary}}}
\AddToHook{env/lemma/begin}{\zcsetup{countertype={thm=lemma}}}
\AddToHook{env/ex/begin}{\zcsetup{countertype={thm=example}}}
\AddToHook{env/rk/begin}{\zcsetup{countertype={thm=remark}}}
\AddToHook{env/Def/begin}{\zcsetup{countertype={thm=definition}}}
\newcommand{\Fq}{\mathbb{F}_{q}}
\newcommand{\FQQ}{\mathbb{F}_{q^{2}}}
\newcommand{\wt}{\mathbf{wt}}

\newcommand{\C}{\mathcal{C}}

\newcommand{\GRS}{\mathrm{GRS}}

\newcommand{\GRL}{\mathrm{GRL}}
\newcommand{\RL}{\mathrm{RL}}
\newcommand{\ba}{\boldsymbol{\alpha}}
\newcommand{\bv}{\boldsymbol{v}}

\newcommand{\perph}{\perp_{h}}
\newcommand{\T}{\!\top}

\begin{document}

\title{Generalized Roth--Lempel Codes: NMDS Characterization, Hermitian Self-Orthogonality, and Quantum Constructions}
\author{Qi~Liu$^{1}$, Xuefei~Wu$^{1}$, Yingchun Cheng$^{1}$, and Haiyan~Zhou$^{1*}$%
  \thanks{$^{1}$School of Mathematical Sciences, Nanjing Normal University, Nanjing 210023, China. (E-mails: QiLiu67@aliyun.com; 240901020@njnu.edu.cn; 240901002@njnu.edu.cn; zhouhy@njnu.edu.cn.)}%
  \thanks{$^{*}$Corresponding author: Haiyan Zhou.}%
  \thanks{This paper is supported by National Natural Science Foundation of China under Grant Nos. 12471493 and 12441105.}%
}
\maketitle
\begin{abstract}
  In their seminal 1989 work (IEEE Trans. Inf. Theory 35(3):655-657), Roth and Lempel constructed a well-known family of non-Reed-Solomon maximum distance separable (MDS) codes. For decades, this family of codes has  attracted extensive research attention due to its algebraic structure,  low-complexity decoding, and broad applications in cryptography and data storage. Most recently, in 2025, the generalized Roth-Lempel (GRL) framework unifies Roth-Lempel codes and its extensions under a flexible algebraic structure. However, explicit criteria for the near-MDS (NMDS) property of GRL codes have not been established, and no systematic construction of Hermitian self-orthogonal GRL codes has been reported, limiting their deployment in classical and quantum error correction.

  In this work, we make three contributions to address these gaps. First,  we give explicit necessary and sufficient  conditions for the NMDS property of the two most widely used subclasses  of GRL codes. Second, we construct four new families of Hermitian self-orthogonal codes from GRL codes. Two of these families are NMDS, with parameters not covered by  existing Hermitian self-orthogonal NMDS codes. Third, based on the  proposed Hermitian self-orthogonal GRL codes, we construct four families of quantum GRL codes, including two infinite families of quantum  NMDS codes that attain the quantum Singleton bound minus one. Compared  to the known quantum error-correcting codes, we obtain many new or improved quantum error-correcting codes. This work bridges the gap  between classical GRL code families and quantum error-correction  applications.
\end{abstract}

%% Index Terms: Roth-Lempel codes; Generalized Roth-Lempel (GRL) codes; Near maximum distance separable (NMDS) codes; Hermitian self-orthogonal codes; Quantum error-correcting codes    

\section{Introduction}
Let $\Fq$ be a finite field with $q$ elements, where $q$ is a prime power. A linear code \(\C\) with parameters \([n,k,d] _{q}\) is a \(k\)-dimensional subspace of \(\Fq ^{n}\), which has minimum Hamming distance \(d\).
The Singleton defect \cite{defect_SGB} of \(\C\) is a measure of how far \(\C\) is from being a maximum distance separable (MDS) code, defined as \(s(\C) = n-k+1-d\). If \(s(\C)=0\), then \(\C\) is an MDS code. If \(s(\C)=1\), then \(\C\) is called an almost MDS (AMDS) code. Let
\[
  \C ^{\perp} = \{\boldsymbol{u} \in \Fq ^{n} \mid \boldsymbol{u} ^{\T} \boldsymbol{c} = 0, \forall \boldsymbol{c} \in \C\}.
\]
A linear code \(\C\) is called a near MDS (NMDS) code if \(s(\C) = s(\C ^{\perp}) = 1\).

These three classes of codes are of great importance in coding theory because of their powerful error-correction capabilities and wide applications in cryptography and distributed storage.
Let \( \alpha_1,\cdots ,\alpha _{n}\) be distinct elements in \( \Fq \cup \{\infty\}\). The Reed-Solomon (RS) code \cite{RS} is defined by
\[
  \{(f(\alpha_1), \cdots ,f(\alpha_n))\mid f(x) \in \Fq[x], \deg f< k\},
\]
where \( f(\infty)\) is defined as the coefficient of \( x ^{k-1}\) in \(f\). To date, the most prominent MDS codes are RS codes and their generalizations, which have been extensively studied in coding theory (readers can refer to \cite{grassl1999quantum,Li_2008,Luo_2019,Jin_2017,Sui_2022,ZHANG_2013}). Against this background, it is natural to ask whether one can systematically construct new code families that retain strong distance properties while offering greater structural flexibility. The construction of MDS codes that are not equivalent to RS codes (shortly, non-RS-type MDS codes) is a long-standing and important problem in both coding theory and finite geometry \cite{Chen_2024,11311495, Roth_1989}. In 1989, Roth and Lempel \cite{Roth_1989} introduced a family of codes \(\RL _{k}(\ba, \delta)\) by extending the generator matrix of RS codes with two columns as follows:
\[
  \begin{pNiceArray}{cc|c}[columns-width=auto,margin]
    \Block{2-2}{\displaystyle G_{\text{RS}}(\ba)} & & O \\
    & & A_\delta
  \end{pNiceArray} _{k \times (n+2)},
\]
where \(G _{\text{RS}}(\ba)\) is the generator matrix of the RS code with the evaluation-point sequence \(\ba = (\alpha_1,\alpha_2,\cdots ,\alpha _{n}) \in \Fq ^{n}\), and \(A _{\delta} = \begin{bmatrix}
  0 & 1      \\
  1 & \delta
\end{bmatrix}\) with \(\delta \in \Fq\).
Roth and Lempel provided necessary and sufficient conditions for Roth-Lempel codes to be MDS and showed that MDS Roth-Lempel codes are non-RS-type MDS codes. However, only in 2023 did Han and Fan \cite{Han_2023} provide a necessary and sufficient condition for \(\RL _{k}(\ba,0)\) to be NMDS. In 2026, Xu and Zhou \cite{11222820} completely determined the NMDS condition for Roth-Lempel codes.

In 2024, Wu et al. \cite{wu2024more} studied an extended class of Roth-Lempel codes. They gave equivalent conditions for this class of codes or its dual to be MDS or AMDS, respectively. Following these results, in 2025, Liang and Liao \cite{liang2025equivalent}, and Li, Jiang, and Liu \cite{Li_2025}, independently introduced generalized Roth-Lempel (GRL) codes, which offer more flexible constructions than Roth-Lempel codes.
Li, Jiang, and Liu \cite{Li_2025} gave necessary and sufficient conditions for two classes of GRL codes to be MDS. Liang and Liao \cite{Liang_2026} constructed two families of NMDS GRL codes. In \cite{LIANG2026115084}, Liang and Liao also considered the extended code of a class of GRL codes and established equivalent conditions for these codes, or their dual codes, to be MDS or AMDS.

A notable advantage of GRL codes is that their duality properties are easier to study than those of Roth-Lempel codes, which is important for quantum error-correction applications. Li, Jiang, and Liu \cite{Li_2025} gave an equivalent condition for Euclidean self-duality for two classes of GRL codes. More recently, Liang, Huang, Liao, Fan, and Zhou \cite{liang2026non} constructed several Euclidean and Hermitian linear complementary dual (LCD) codes using GRL codes. However, the Hermitian self-orthogonal GRL codes is still unknown. This missing piece is critical because Hermitian self-orthogonality is a standard gateway to quantum code constructions. More broadly, there has been little work on the Hermitian self-orthogonality of NMDS codes. In 2014, Suprijanto, Renata, and Putri \cite{2014NMDS} constructed many Hermitian self-dual NMDS codes over \( \mathbb{F} _{9}, \mathbb{F} _{25}, \mathbb{F} _{121}\) of length up to \(14\). Guo, Li, Liu, and Song \cite{Guo_2022} constructed a family of Hermitian self-dual NMDS codes over \( \mathbb{F}_{q ^{2}}\) by using generalized twisted Reed-Solomon codes with length less than \(q+1\). Zhu and Wan \cite{Zhu_2025} continued this line of work on Galois self-dual GRS and twisted GRS (TGRS) codes and presented a family of Hermitian self-dual TGRS codes that are NMDS.

To the best of the authors' knowledge, no systematic constructions of Hermitian self-orthogonal GRL codes and quantum GRL codes have been reported. Accordingly, the classical structural theory and quantum applications of GRL codes remain largely unexplored.

\textbf{Contributions of this work.} To fill these gaps, we first completely determine the NMDS properties for two classes: GRL codes with \(s=2\) and \(s=3\) (cf. \zcref{def:GRL}). We then establish equivalent conditions for Hermitian self-orthogonality for these two classes of GRL codes. Finally, based on these conditions, we construct four families of Hermitian self-orthogonal GRL codes with flexible parameters. Specifically, we provide Hermitian self-orthogonal GRL codes with the following parameters:
\begin{enumerate}
  \item \( [m(q+1)+2, k, m(q+1)-k] _{q ^{2}}\) for \( q\geq 4\), \( 2\leq m\leq q-2\), and \( 3\leq k\leq  m+1\);
  \item \( [m(q-1)+2, k, m(q-1)+2-k] _{q ^{2}}\) for \( q\geq 5\), \( 2\leq  m\leq q+1\), and \( 3\leq  k\leq  \frac{q+1}{2}\);
  \item \( [m(q+1)+3, k, \geq m(q+1)-k+2] _{q ^{2}}\) for \( q\geq 5\), \( 2\leq  m\leq  q-3\), and \( 4\leq  k\leq m+2\);
  \item \( [m(q-1)+3,k, \geq  m(q-1)-k+2] _{q ^{2}}\) for \( q\geq 7\), \( 2\leq m\leq  q+1\), and \( 4\leq k\leq \frac{q+1}{2}+1\).
\end{enumerate}
The first two families are NMDS codes, and their parameters are not covered by existing Hermitian self-orthogonal NMDS codes in \cite{Guo_2022, 2014NMDS, Zhu_2025}.
Based on these classical constructions, we then construct quantum GRL codes, two families of which are quantum NMDS codes. Furthermore, Table II presents numerous new or improved quantum codes with a smaller Singleton defect compared to the known quantum codes in \cite{BE:tables, table2, table3, table4}.

\textbf{The organization of this paper is as follows.} In Section 2, we introduce basic notions, definitions, and tools related to GRL codes. On this basis, in Section 3, we establish NMDS criteria for two classes of GRL codes. Next, in Section 4, we construct four families of Hermitian self-orthogonal GRL codes. Finally, in Section 5, we construct quantum GRL codes by applying these Hermitian self-orthogonal constructions.

\section{Preliminaries}
For a prime power \(q\), let \(n\) be an integer with \(n \leq q\). Let \(\boldsymbol{\alpha}=\{\alpha_1,\ldots,\alpha_n\}\) be a set of \(n\) pairwise distinct elements of \(\Fq\), and let \(\boldsymbol{v}=(v_1,\ldots,v_n) \in (\Fq^*)^n\). For any integer \(k\) such that \(1 \leq k \leq n\), we define the \(k \times n\) matrix \(G_k\) as
\[
  G_{k} = \begin{pmatrix}
    v_1               & v_2               & \cdots & v_n               \\
    v_1\alpha_1       & v_2\alpha_2       & \cdots & v_n\alpha_n       \\
    \vdots            & \vdots            & \ddots & \vdots            \\
    v_1\alpha_1^{k-1} & v_2\alpha_2^{k-1} & \cdots & v_n\alpha_n^{k-1}
  \end{pmatrix}. \tag{MGRS}\label{eq:MGRS}
\]

\begin{Def}\label{def:GRL}
  Let \(A_s\) be an \(s \times s\) nonsingular matrix over \(\Fq\), and let \(\boldsymbol{\alpha}\) and \(\boldsymbol{v}\) be as previously defined. For integers \(k\) and \(s\) satisfying \(s \leq k \leq n \leq q\), we define a code \(\GRL_{k}(\boldsymbol{\alpha},\boldsymbol{v},A_s)\) as a generalized Roth-Lempel type code if its generator matrix is
  \[
    G_{\GRL}=
    \begin{pNiceArray}{cc|c}[columns-width=auto,margin]
      \Block{2-2}{\displaystyle G_{k}} & & O \\
      & & A_s
    \end{pNiceArray}
  \]
  where \(G_{k}\) is the \(k \times n\) matrix defined in \eqref{eq:MGRS}, and \(O\) is the \((k-s) \times s\) zero matrix.
\end{Def}

As a well-known result, the linear code generated by \(G_k\) is an \([n,k,n-k+1]\) MDS code, which is called a generalized Reed-Solomon (GRS) code and denoted by \(\GRS_{k}(\boldsymbol{\alpha},\boldsymbol{v})\). More details about GRS codes can be found in \cite{AFFC}.

By the definition of \(\GRL _{k}(\ba,\bv,A_s)\), we have \(\dim(\GRL _{k}(\ba,\bv,A_s))=k\) and length \(n+s\). It is easy to see that if we take \( s=1\), \( A_1 = (u) \) with \(u \in \Fq^*\), then the \( \GRL _{k}(\ba,\bv,A_1)\) is exactly the Extended GRS code defined in \cite{FDEC}. And, when we take \( s =2\), and \(A_2 = \begin{bmatrix}
  0 & 1      \\
  1 & \delta
\end{bmatrix}\), \(\GRL _{k}(\ba,\bv,A_2)\) is exactly the Roth-Lempel code \(\RL _{k}(\ba,\delta)\).

The following lemma is a key result about the NMDS property of linear codes, which is useful to determine the NMDS property for GRL codes.

\begin{lemma}\label{lem:NMDS}\cite{Dodunekov_1995}
  Let \(\C\) be a linear code over \( \mathbb{F}_{q}\) with parameters \([n,k]_{q}\). Then \(\C\) is an NMDS code if and only if the following conditions hold:
  \begin{enumerate}[label=(\roman*)]
    \item Any \(k-1\) columns of a generator matrix of \(\C\) are linearly independent;
    \item There exist \(k\) columns of a generator matrix of \(\C\) that are linearly dependent;
    \item Any \(k+1\) columns of generator matrix of \(\C\) have rank \(k\).
  \end{enumerate}
\end{lemma}

\zcref{lem:NMDS} also has a projective hyperplane version. Readers can refer to \cite{Han_2023} for more details.

Another helpful result comes from the study of the subset sum problem (SSP). Let us recall some necessary notation for SSP. For a subset \(S \subset \Fq\) and an element \(\delta \in \Fq\), we define
\[
  N(t,\delta,S) = \#\{I \subset S \mid |I| = t, \sum_{i \in I} i = \delta\}, \quad 1\leq t \leq |S|.
\]
For a subset \( I\) of \( S\) with \(|I| = t\), we call \(I\) a \(t\)-subset of \(S\). The fundamental problem of SSP is to determine whether \(N(t,\delta,S) > 0\) for given \(t\), \(\delta\), and \(S\). The following two lemmas provide necessary and sufficient conditions for \(N(t,\delta,S) = 0\) when \(S = \Fq\) or \(S = \Fq^{*}\), where \(q\) is a power of a prime \(p\).

\begin{lemma}\label{lem:sum1}\cite[Theorem 2]{ 11222820}
  Let \(\delta \in \Fq\), and let \(t\) be a positive integer with \(t\leq q\). Then \(N(t,\delta, \Fq) = 0\) if and only if \(p|t\) and one of the following holds:
  \begin{itemize}
    \item \(p\) is odd, and \(t=p=q\);
    \item \(p\) is odd, and \(p<t =q\);
    \item \(p=2\), and \(p<t=q\).
  \end{itemize}
\end{lemma}

\begin{lemma}\label{lem:sum2}\cite[Theorem 3]{ 11222820}
  Let \(\delta \in \Fq \), and let \(t\) be a positive integer with \(t\leq q\). Then \(N(t,\delta,\Fq ^{*}) = 0\) if and only if one of the following holds:
  \begin{itemize}
    \item \(p=2\), \(t=q-1\), and \(q>p\);
    \item \(p\) is odd, and \(t = q-1\).
  \end{itemize}
\end{lemma}

At the end of this section, we recall some definitions of Hermitian duality of linear codes. For any two vectors \(\bv = (v_1,\cdots ,v_n)\) and \(\boldsymbol{u} = (u_1,\cdots ,u_n)\) in \(\FQQ ^{n}\), the Hermitian inner product of \(\bv\) and \(\boldsymbol{u}\) is defined as
\[
  \langle \bv,  \boldsymbol{u} \rangle _{H} = \sum_{i=1} ^{n} v_i \overline{u_i},
\]
where \( \overline{(\cdot)}\) is the Frobenius automorphism of \(\FQQ\) defined by \(\overline{x} = x ^{q}\) for any \(x \in \FQQ\).
For any \(x \in \FQQ\), let \(N(x):= x ^{q+1}\). Then it is easy to see \(\langle \bv, \bv \rangle _{H}  = \sum _{i=1} ^{n} N(v_i)\). For a linear code \(\C\) over \(\FQQ\), the Euclidean dual code of \(\C\) is defined as
\[
  \C ^{\perp} := \{\boldsymbol{u} \in \FQQ ^{n} \mid \boldsymbol{u} ^{\T} \boldsymbol{c} = 0, \forall \boldsymbol{c} \in \C\},
\]
and the Hermitian dual code of \(\C\) is defined as
\[
  \C ^{\perph} := \{\boldsymbol{u} \in \FQQ ^{n} \mid \langle \boldsymbol{u}, \boldsymbol{c} \rangle _{H} = 0, \forall \boldsymbol{c} \in \C\}.
\]
The two type of dual codes are usually different, but they have the same parameters by the following well-known result:
\[
  \C ^{\perph} = \left(\C ^{\perp}\right) ^{q} = \{\boldsymbol{u} ^{q} \mid \boldsymbol{u} \in \C ^{\perp}\}.
\]
So if a linear code is a NMDS code respect to the Euclidean duality, then it is also a NMDS code respect to the Hermitian duality.

We say a linear code \(\C\) is Hermitian self-orthogonal if \(\C \subset \C ^{\perph}\) and Hermitian self-dual if \(\C = \C ^{\perph}\). The Hermitian orthogonal codes are important ingredients for the construction of quantum error-correcting codes. We will construct quantum GRL codes by applying the Hermitian self-orthogonal construction in Section 5.

\section{NMDS criterion for the generalized Roth-Lempel type codes}

In this section, we establish NMDS criteria for two classes of GRL codes. We begin with GRL codes with \(s=2\), which contain Roth-Lempel codes as a special case, and then move to GRL codes with \(s=3\), which generalize the Roth-Lempel-type codes in \cite{wu2024more}. This progression from \(s=2\) to \(s=3\) will also support the constructions in the subsequent sections.

\begin{thm}\label{NMDS2}
  Let \(A_2 = \begin{bmatrix}
    a & b \\
    c & d
  \end{bmatrix}\) be a \(2 \times 2 \) non-singular matrix over \(\Fq\), \(\ba = \{\alpha_1,\cdots ,\alpha_n\} \subset \Fq\) be \(n\) distinct elements, \(\bv \in (\Fq ^{*}) ^{n}\). Then \(\GRL_{k}(\ba,\bv,A_2)\) is a NMDS code if and only if at least one of the following holds:
  \begin{itemize}
    \item \(\exists \sigma _{k-1} \in \Delta _{k-1}\) such that \(c= a \sigma _{k-1} \);
    \item \(\exists \sigma _{k-1} \in \Delta _{k-1}\) such that \(d = b \sigma _{k-1} \),
  \end{itemize}
  where \(\Delta _{k-1} = \left\{\sum _{e \in I} e \mid \forall I \subset \ba \text{ and } |I| = k-1 \right\}\).
\end{thm}
\begin{proof}

  By \zcref{lem:NMDS}, we need check three conditions for the generator matrix \(G_{\GRL}\) of \(\GRL_{k}(\ba,\bv,A_2)\). Let \[ \{g_1,g_2,\cdots ,g_n, g _{n+1}, g _{n+2}\}\] be the columns of \(  G _{\GRL}\), where \(g_i\) is the \(i\)-th column of \(G_k\) for \(1 \leq i \leq n\), and \(g_{n+1}\) and \(g_{n+2}\) are the last two columns of \(G_{\GRL}\).

  \textbf{Condition (i)}: If \( k=2\), condition (i) of \zcref{lem:NMDS} is trivial. We assume \( k\geq 3\). Let \(G\) be a matrix consisting of \(k-1\) columns from \(G_{\GRL}\). We consider three cases:
  \begin{itemize}
    \item If the columns of \(G\) do not include \(g_{n+1}\) or \(g_{n+2}\), then \(G\) generates a GRS code and thus has full rank \(k-1\).
    \item If the columns of \(G\) contain both \(g_{n+1}\) and \(g_{n+2}\), let \(M\) be a minor of \(G\) formed by the first \(k-3\) rows and both of the last two rows. Then \(M\) is the product of \(\det(A_2) \neq 0\) and the determinant of a \((k-2) \times (k-2)\) Vandermonde determinant, which is non-zero. Thus, \(G\) has full rank \(k-1\).
    \item If the columns of \(G\) contain exactly one of \(g_{n+1}\) or \(g_{n+2}\), without loss of generality, we assume
          \[
            G = \begin{pmatrix}
              v_{i_1}                   & v_{i_2}                   & \cdots & v_{i_{k-2}}                       & 0      \\
              v_{i_1}\alpha_{i_1}       & v_{i_2}\alpha_{i_2}       & \cdots & v_{i_{k-2}}\alpha_{i_{k-2}}       & 0      \\
              \vdots                    & \vdots                    & \ddots & \vdots                            & \vdots \\
              v_{i_1}\alpha_{i_1}^{k-2} & v_{i_2}\alpha_{i_2}^{k-2} & \cdots & v_{i_{k-2}}\alpha_{i_{k-2}}^{k-2} & a      \\
              v_{i_1}\alpha_{i_1}^{k-1} & v_{i_2}\alpha_{i_2}^{k-1} & \cdots & v_{i_{k-2}}\alpha_{i_{k-2}}^{k-1} & c
            \end{pmatrix}.
          \]

          When \(a \neq 0\), the minor

          \[
            M_k = \begin{vmatrix}
              v_{i_1}                   & v_{i_2}                   & \cdots & v_{i_{k-2}}                       & 0      \\
              v_{i_1}\alpha_{i_1}       & v_{i_2}\alpha_{i_2}       & \cdots & v_{i_{k-2}}\alpha_{i_{k-2}}       & 0      \\
              \vdots                    & \vdots                    & \ddots & \vdots                            & \vdots \\
              v_{i_1}\alpha_{i_1}^{k-2} & v_{i_2}\alpha_{i_2}^{k-2} & \cdots & v_{i_{k-2}}\alpha_{i_{k-2}}^{k-2} & a
            \end{vmatrix} \neq 0.
          \]

          When \(c \neq 0\), the minor

          \[
            M_{k-1} = c \cdot \begin{vmatrix}
              v_{i_1}                   & v_{i_2}                   & \cdots & v_{i_{k-2}}                       \\
              v_{i_1}\alpha_{i_1}       & v_{i_2}\alpha_{i_2}       & \cdots & v_{i_{k-2}}\alpha_{i_{k-2}}       \\
              \vdots                    & \vdots                    & \ddots & \vdots                            \\
              v_{i_1}\alpha_{i_1}^{k-3} & v_{i_2}\alpha_{i_2}^{k-3} & \cdots & v_{i_{k-2}}\alpha_{i_{k-2}}^{k-3}
            \end{vmatrix} \neq 0.
          \]
          Since \(A_2\) is non-singular, its columns cannot be zero vectors, so \(a\) and \(c\) cannot be simultaneously zero. Thus \(G\) always has a non-zero minor of order \(k-1\).
  \end{itemize}
  Therefore, condition (i) always holds.

  \textbf{Condition (ii)}: Similarly, let \(G'\) be a \(k \times k\) submatrix formed by selecting \(k\) columns from \(G_{\GRL}\). Following the proof of condition (i), we consider three cases based on the inclusion of columns \(g_{n+1}\) and \(g_{n+2}\).

  First, if the columns of \(G'\) do not include \(g_{n+1}\) or \(g_{n+2}\), then \(G'\) is a generator matrix of a GRS code of dimension \(k\). Thus, \(\det(G') \neq 0\).

  Second, if \(G'\) contains both \(g_{n+1}\) and \(g_{n+2}\), without loss of generality, we arrange \( g_{n+1}\) and \( g _{n+2}\) as the last two columns of \(G'\). By Laplace expansion along the last two columns of \( G'\), \(\det(G')\) is the product of \(\det(A_2) \neq 0\) and the determinant of a \((k-2) \times (k-2)\) generator matrix of a GRS code. Since both factors are non-zero, we have \(\det(G') \neq 0\).

  Consequently, condition (ii) holds if and only if \(\det(G') = 0\) for some matrices \(G'\) that contain exactly one of \( \{g_{n+1}, g_{n+2}\}\). Without loss of generality, let
  \begin{align*}
    \det(G') & = \begin{vmatrix}
                   v_{i_1}                   & v_{i_2}                   & \cdots & v_{i_{k-1}}                       & 0      \\
                   v_{i_1}\alpha_{i_1}       & v_{i_2}\alpha_{i_2}       & \cdots & v_{i_{k-1}}\alpha_{i_{k-1}}       & 0      \\
                   \vdots                    & \vdots                    & \ddots & \vdots                            & \vdots \\
                   v_{i_1}\alpha_{i_1}^{k-2} & v_{i_2}\alpha_{i_2}^{k-2} & \cdots & v_{i_{k-1}}\alpha_{i_{k-1}}^{k-2} & a      \\
                   v_{i_1}\alpha_{i_1}^{k-1} & v_{i_2}\alpha_{i_2}^{k-1} & \cdots & v_{i_{k-1}}\alpha_{i_{k-1}}^{k-1} & c
                 \end{vmatrix}           \\[4mm]
             & = c \begin{vmatrix}
                     v_{i_1}                   & v_{i_2}                   & \cdots & v_{i_{k-1}}                       \\
                     v_{i_1}\alpha_{i_1}       & v_{i_2}\alpha_{i_2}       & \cdots & v_{i_{k-1}}\alpha_{i_{k-1}}       \\
                     \vdots                    & \vdots                    & \ddots & \vdots                            \\
                     v_{i_1}\alpha_{i_1}^{k-3} & v_{i_2}\alpha_{i_2}^{k-3} & \cdots & v_{i_{k-1}}\alpha_{i_{k-1}}^{k-3} \\
                     v_{i_1}\alpha_{i_1}^{k-2} & v_{i_2}\alpha_{i_2}^{k-2} & \cdots & v_{i_{k-1}}\alpha_{i_{k-1}}^{k-2}
                   \end{vmatrix}
    - a \begin{vmatrix}
          v_{i_1}                   & v_{i_2}                   & \cdots & v_{i_{k-1}}                       \\
          v_{i_1}\alpha_{i_1}       & v_{i_2}\alpha_{i_2}       & \cdots & v_{i_{k-1}}\alpha_{i_{k-1}}       \\
          \vdots                    & \vdots                    & \ddots & \vdots                            \\
          v_{i_1}\alpha_{i_1}^{k-3} & v_{i_2}\alpha_{i_2}^{k-3} & \cdots & v_{i_{k-1}}\alpha_{i_{k-1}}^{k-3} \\
          v_{i_1}\alpha_{i_1}^{k-1} & v_{i_2}\alpha_{i_2}^{k-1} & \cdots & v_{i_{k-1}}\alpha_{i_{k-1}}^{k-1}
        \end{vmatrix}                            \\[4mm]
             & = \prod_{t=1}^{k-1} v_{i_t} \cdot \prod_{1\leq p < q \leq k-1} (\alpha_{i_q} - \alpha_{i_p}) \cdot \left[ c - a \left( \sum_{t=1}^{k-1} \alpha_{i_t} \right) \right].
  \end{align*}
  It follows that \(\det(G') = 0\) if and only if \(c- a \cdot \left(\sum _{t=1}^{k-1} \alpha _{i _{t}}\right) = 0\).

  \textbf{Condition (iii)}: Let \(G''\) be a matrix formed by \(k+1\) columns of \(G _{\GRL}\). If \( g _{n+1}\), and \( g _{n+2}\) are not columns of \( G''\), then any minor of \(G''\) with order \(k\) is nonzero. If exactly one of \( \{g _{n+1}, g _{n+2}\}\) is a column of \(G''\), then the first \(k\) columns of \(G ''\) is a Vandermonde matrix. Finally, if the columns of \(G''\) contains \( g _{n+1}, g _{n+2}\), we assume
  \[
    G'' = \begin{pmatrix}
      v_{i_1}                   & v_{i_2}                   & \cdots & v_{i_{k-1}}                       & 0      & 0      \\
      v_{i_1}\alpha_{i_1}       & v_{i_2}\alpha_{i_2}       & \cdots & v_{i_{k-1}}\alpha_{i_{k-1}}       & 0      & 0      \\
      \vdots                    & \vdots                    & \ddots & \vdots                            & \vdots & \vdots \\
      v_{i_1}\alpha_{i_1}^{k-2} & v_{i_2}\alpha_{i_2}^{k-2} & \cdots & v_{i_{k-1}}\alpha_{i_{k-1}}^{k-2} & a      & b      \\
      v_{i_1}\alpha_{i_1}^{k-1} & v_{i_2}\alpha_{i_2}^{k-1} & \cdots & v_{i_{k-1}}\alpha_{i_{k-1}}^{k-1} & c      & d
    \end{pmatrix}.
  \]

  It is easy to check that the following minor is not zero:

  \[
    \begin{vmatrix}
      v_{i_2}                   & v_{i_3}                   & \cdots & v_{i_{k-1}}                       & 0      & 0      \\
      v_{i_2}\alpha_{i_2}       & v_{i_3}\alpha_{i_3}       & \cdots & v_{i_{k-1}}\alpha_{i_{k-1}}       & 0      & 0      \\
      \vdots                    & \vdots                    & \ddots & \vdots                            & \vdots & \vdots \\
      v_{i_2}\alpha_{i_2}^{k-2} & v_{i_3}\alpha_{i_3}^{k-2} & \cdots & v_{i_{k-1}}\alpha_{i_{k-1}}^{k-2} & a      & b      \\
      v_{i_2}\alpha_{i_2}^{k-1} & v_{i_3}\alpha_{i_3}^{k-1} & \cdots & v_{i_{k-1}}\alpha_{i_{k-1}}^{k-1} & c      & d
    \end{vmatrix}.
  \]
  Therefore, by the arbitrary of choice of \( G''\), any \(k+1\) columns of \(G _{\GRL}\) have rank \(k\).
\end{proof}

If we take
\[
  A_2 = A_\delta = \begin{bmatrix}
    0 & 1      \\
    1 & \delta
  \end{bmatrix},
\]
with \(\delta \in \Fq\) in \zcref{NMDS2}, then we have the following corollary.
\begin{cor}
  Let \(\ba\) be an \( n\)-subset of \( \Fq\) and \(\delta \in \Fq\). Then \(\RL_{k}(\ba,A_\delta)\) is a NMDS code if and only if there exists a \(\sigma _{k-1} \in \Delta _{k-1}\) such that \(\delta = \sigma _{k-1}\).
\end{cor}

This Corollary can also be obtained by considering the number of zeros of the projective hyperplane. The reader can refer to \cite{11222820}.

\begin{rk}
  In the original definition of Roth-Lempel codes, the dimension of Roth-Lempel code greater than 2 which is necessary for the view of projective geometry. By the same reason, \cite{11222820} also miss the case of \( k=2\) in their result. However, in our way of proof, the case of \( k=2\) can be easy to check.
\end{rk}

\begin{ex}
  Let \( p =3\), \( \Fq = \mathbb{F} _{3}(\omega)\) with \( \omega ^{2} =2\). We set
  \[
    \ba = \{\omega, 2\omega, 1, \omega+2, 2\omega+2\}, \quad \bv = (1,1,1,1,1), \quad \text{and } A_2 = \begin{bmatrix}
      1      & 1 \\
      \omega & 2
    \end{bmatrix}.
  \]
  Then \( \GRL _{2}(\ba,\bv, A_2)\) is a NMDS with parameters \( [7,2,5]_{9}\),
\end{ex}

Combining \zcref{lem:sum1} and \zcref{lem:sum2} with the above theorem, we have the following corollary.

\begin{cor}
  \begin{enumerate}
    \item For \(2\leq k\leq q\), \(\GRL_{k}(\Fq,\bv,A_2)\) is an NMDS code with parameters \([q+2,k,q+2-k]\) for all \(\bv\in(\Fq^{*})^{n}\) and for any choice of a nonsingular \(2\times 2\) matrix \(A_2\) over \(\Fq\).
    \item For \(2\leq k\leq q-2\), \(\GRL_{k}(\Fq^{*},\bv,A_2)\) is an NMDS code with parameters \([q+1,k,q+1-k]\) for all \(\bv\in(\Fq^{*})^{n}\) and for any choice of a nonsingular \(2\times 2\) matrix \(A_2\) over \(\Fq\).
  \end{enumerate}
\end{cor}
\begin{proof}
  When \( k=2\), the result follows directly by \zcref{NMDS2}. Without loss of generality, we assume \(k\geq 3\). Let \[
    A_2 = \begin{bmatrix}
      a & b \\
      c & d
    \end{bmatrix}.
  \]
  If \(a\neq 0\), we just need to find a \(\sigma \in \Delta _{k-1}\) such that \(\sigma = \frac{c}{a}\). On the contrary, if \(a = 0\), then \(c \neq 0\) and \(b\neq 0\). There has two subcases:
  \begin{itemize}
    \item If \(d = 0\), we just need to find a \(\sigma \in \Delta _{k-1}\) such that \(\sigma = \frac{d}{b} = 0\).
    \item If \(d \neq 0\), we just need to find a \(\sigma \in \Delta _{k-1}\) such that \(\sigma = \frac{d}{b}.\)
  \end{itemize}
  By \zcref{lem:sum1} and \zcref{lem:sum2}, the above conditions hold for any \(a,b,c,d\) if \(n=q\) and \(n=q-1\), respectively. Hence, the corollary holds.

\end{proof}

\begin{rk}

  \begin{enumerate}
    \item By the above corollary, it is directly to get two classes of NMDS codes obtained in \cite{Liang_2026}.
    \item The condition of \zcref{NMDS2} is equivalent to the following condition:
          \[ N(k-1, \frac{c}{a}, \ba)\neq 0 \text{ if } a\neq 0 \text{ or }N(k-1, \frac{d}{b},\ba)\neq 0 \text{ if } b\neq 0 .\]
  \end{enumerate}

\end{rk}
The class of GRL codes with \(s=2\) contains the classical Roth-Lempel codes. We now turn to the \(s=3\) case, which generalizes the Roth-Lempel-type codes defined in \cite{wu2024more}. The following lemma will be repeatedly used in the remainder of the paper.

For \(n\) distinct elements of \(\Fq\), denoted by \(\ba = \{\alpha_1,\alpha_2,\cdots ,\alpha _{n}\}\). Let
\[
  w_i = \prod _{j=1, j\neq i} ^{n} \frac{1}{\alpha_i-\alpha _{j}}.
\]
We define \(R _{i}=(\alpha_1 ^{i}, \alpha_2 ^{i},\cdots ,\alpha _{n}^{i})\). Then \(\{R_0,\cdots ,R _{n-1}\}\) is a basis of \(\Fq ^{n}\). For \(t\geq  1\), let
\[
  R _{n-1+t} = \sum _{i=0} ^{n-1} \gamma _{i}(t) R_i.
\]

\begin{lemma}\label{solution}\cite{Sui_2022}
  \[
    \sum _{i=1}^{n} w_i \alpha _{i}^{n-1+t} = \begin{cases}
      0,                & \text{ if }t<0 \\
      1,                & \text{ if }t=0 \\
      \gamma _{n-1}(t), & \text{ if }t>0
    \end{cases}.
  \]
\end{lemma}

\begin{rk}
  It is easy to check
  \[
    \gamma _{n-1}(1) = \sum _{i=1} ^{n} \alpha _{i},
    \quad \gamma _{n-1}(2) = \sum _{i=1} ^{n} \alpha_i ^{2}+ \sum _{1\leq i<j\leq n} \alpha _{i} \alpha _{j}.
  \]
\end{rk}

For \(n\) distinct elements of \(\Fq\), denoted by \(\ba = \{\alpha_1,\alpha_2,\cdots ,\alpha _{n}\}\), we define the following two sets:
\begin{align*}
  \Omega _{k}(\ba) & = \left \{ b _{I} = (1, -\sigma_1(I), \sigma_1 ^{2}(I)-\sigma_2(I)) ^{\T} \mid \forall I \subset \ba \text{ and } |I| = k \right\}, \\
  \Gamma _{k}(\ba) & = \left\{b _{I} = (\sigma _{2}(I), \sigma _{1}(I), 1 ) ^{\T} \mid \forall I \subset \ba \text{ and } |I| = k \right\},
\end{align*}
where
\[
  \sigma _{1}(I) = -\sum _{i\in I} \alpha _{i}, \quad \sigma _{2}(I) = \sum _{\substack{i,l\in I \\ i<l}} \alpha _{i} \alpha _{l}.
\]

\begin{thm}\label{NMDS3}
  Let \(A_3 = (\mathbf{a}_1,\mathbf{a}_2,\mathbf{a}_3)\) be a \(3 \times 3\) nonsingular matrix over \(\Fq\), \(n\) an integer with \(n \leq q\), \(\ba\) an \(n\)-subset of \(\Fq\), and \(\bv \in (\Fq^*)^n\). Then \(\GRL_k(\ba,\bv, A_3)\) is an NMDS code if and only if the following three conditions hold:
  \begin{enumerate}
    \item At least one of the following holds:
          \begin{itemize}
            \item[(a)] There exist distinct indices \(i, j \in \{1,2,3\}\) and a vector \(b_I \in \Omega_{k-2}(\ba)\) such that \(\det(\mathbf{a}_i, \mathbf{a}_j, b_I) = 0\);
            \item[(b)] There exists an index \(i \in \{1,2,3\}\) and a vector \(b_J \in \Gamma_{k-1}(\ba)\) such that \(\mathbf{a}_i^T b_J = 0\).
          \end{itemize}
    \item For every \(b_J \in \Gamma_{k-1}(\ba)\), there is at most one \(i \in \{1,2,3\}\) such that \(\mathbf{a}_i^T b_J = 0\).
    \item For every \(j \in \{1,2,3\}\) and every \(b_I \in \Omega_{k-2}(\ba)\), \(\mathbf{a}_j \notin \langle b_I \rangle\).
  \end{enumerate}
\end{thm}
\begin{proof}
  As the proof of Theorem 1, we need to check three conditions in \zcref{lem:NMDS} for the generator matrix \(G_{\GRL}\) of \(\GRL _{k}(\ba,\bv, A_3)\). Let \[ \{g_1,\cdots ,g_n, g_{n+1}, g_{n+2}, g_{n+3}\}\] be the columns of \(G_{\GRL}\), where \( g _{i} = (v_i, v_i \alpha _{i}, \cdots , v_i \alpha _{i} ^{k-1}) ^{\T}\) for \( 1\leq i\leq n\) and \(g _{n+j} = (\mathbf{0}, b_j ^{\T})\), \( j=1,2,3\). We define \(S _{\ell}\) as a subset of columns of \(G_{\GRL}\) with size \(\ell\), and put
  \[ t _{\ell} = |S _{\ell} \cap \{g _{n+1}, g_{n+2}, g_{n+3}\}|.\]

  \textbf{Condition (i)}: We choose a subset \( S _{k-1}\). If \(t _{k-1}=0\),  then the vectors of \(S _{k-1}\) are linearly independent since all vectors in \(S _{k-1}\) form a generator matrix of GRS.

  If \( t _{k-1} = 2\), we assume \( S _{k-1}= \{g _{i_1}, g_{i_2},\cdots ,g _{i _{k-3}}, g_{n+1}, g_{n+2}\}\), where \(1\leq i_1< i_2<\cdots <i _{k-3}\leq n\). Since \(A_3\) is a nonsingular matrix, there exists a minor of order 2 of the matrix \([b_1\quad b_2]\) is nonzero. Assume this nonzero minor is the first two rows, then the following matrix has nonzero determinant:
  \[
    \begin{pmatrix}
      v _{i_1}                      & v _{i_2}                      & \cdots & v _{i_{k-3}}                          & 0      & 0      \\
      v _{i_1} \alpha _{i_1}        & v _{i_2} \alpha _{i_2}        & \cdots & v _{i_{k-3}} \alpha _{i_{k-3}}        & 0      & 0      \\
      \vdots                        & \vdots                        & \ddots & \vdots                                & \vdots & \vdots \\
      v _{i_1} \alpha _{i_1} ^{k-3} & v _{i_2} \alpha _{i_2} ^{k-3} & \cdots & v _{i_{k-3}} \alpha _{i_{k-3}} ^{k-3} & a_{11} & a_{12} \\ \\
      v _{i_1} \alpha _{i_1} ^{k-2} & v _{i_2} \alpha _{i_2} ^{k-2} & \cdots & v _{i_{k-3}} \alpha _{i_{k-3}} ^{k-2} & a_{21} & a_{22} \\
    \end{pmatrix}.
  \]
  This minor evaluates to the product of the non-zero minor of \([b_1 \quad b_2]\) and the determinant of a \((k-3)\times (k-3)\) generator matrix of a GRS code, which is nonzero. On the other hand, if \(t _{k-1} = 3\), the following matrix has nonzero determinant:
  \[
    \begin{pmatrix}
      v _{i_1}\alpha _{i_1}         & v _{i_2}\alpha _{i_2}         & \cdots & v _{i _{k-4}}\alpha _{i _{k-4}}         & 0      & 0      & 0      \\
      \vdots                        & \vdots                        & \ddots & \vdots                                  & \vdots & \vdots & \vdots \\ \\
      v _{i_1} \alpha _{i_1} ^{k-3} & v _{i_2} \alpha _{i_2} ^{k-3} & \cdots & v _{i _{k-3}} \alpha _{i _{k-4}} ^{k-3} & a_{11} & a_{12} & a_{13} \\ \\
      v _{i_1} \alpha _{i_1} ^{k-2} & v _{i_2} \alpha _{i_2} ^{k-2} & \cdots & v _{i _{k-4}} \alpha _{i _{k-4}} ^{k-2} & a_{21} & a_{22} & a_{23} \\ \\
      v _{i_1} \alpha _{i_1} ^{k-1} & v _{i_2} \alpha _{i_2} ^{k-1} & \cdots & v _{i _{k-4}} \alpha _{i _{k-4}} ^{k-1} & a_{31} & a_{32} & a_{33}
    \end{pmatrix}.
  \]

  Finally, if \( t _{k-1} = 1\), we assume
  \[
    S_{k-1} = \{g _{i_1},g _{i_2},\cdots ,g _{i _{k-2}}, g _{n+1}\},
  \]
  where \( 1\leq i_1 < i_2 < \cdots < i_{k-1} \leq n\). As \(A_3\) is a nonsingular matrix, we assume \(b_1 = (a_{11},a_{21},a_{31})^{\T}\). All vectors in \(S _{k-1}\) are linearly dependent if and only if there exists a nonzero vector \(\mathbf{c}=(c_1,\cdots ,c _{k-2})^{\T} \in \Fq ^{k-2}\) such that
  \[
    \begin{pmatrix}
      v_{i_1}                   & v_{i_2}                   & \cdots & v_{i_{k-2}}                       \\
      v_{i_1}\alpha_{i_1}       & v_{i_2}\alpha_{i_2}       & \cdots & v_{i_{k-2}}\alpha_{i_{k-2}}       \\
      \vdots                    & \vdots                    & \ddots & \vdots                            \\
      v_{i_1}\alpha_{i_1}^{k-3} & v_{i_2}\alpha_{i_2}^{k-3} & \cdots & v_{i_{k-2}}\alpha_{i_{k-2}}^{k-3} \\
      v_{i_1}\alpha_{i_1}^{k-2} & v_{i_2}\alpha_{i_2}^{k-2} & \cdots & v_{i_{k-2}}\alpha_{i_{k-2}}^{k-2} \\
      v_{i_1}\alpha_{i_1}^{k-1} & v_{i_2}\alpha_{i_2}^{k-1} & \cdots & v_{i_{k-2}}\alpha_{i_{k-2}}^{k-1}
    \end{pmatrix}
    \mathbf{c} = \begin{pmatrix}
      0 \\ 0 \\ \vdots \\ 0 \\ a_{11} \\ a_{21} \\ a_{31}
    \end{pmatrix}.
  \]
  Moreover, as the first \(k-2\) rows of the upper left matrix is a Vandermonde matrix, \(\mathbf{c}\) is unique determined by first \(k-2\) rows. If \( a_{11} = 0\), then \( \mathbf{c}\) must be zero vector, which implies \( S _{k-1}\) has rank \( k-1\). Now, we assume \( a_{11} \neq 0\).
  By \zcref{solution}, we have
  \[
    \mathbf{c} = a_{11} \cdot(\frac{w_1}{v_1}, \frac{w_2}{v_2},\cdots , \frac{w _{k-2}}{v _{k-2}})^{\T}.
  \]
  By \zcref{solution}, it follows that
  \[
    a_{11}\cdot\gamma _{k-3}(1) = a_{11} \cdot (-\sigma_1(I_1))= a_{21}, \qquad a_{11}\cdot \gamma _{k-3}(2)= a_{11} \cdot (\sigma_1(I_1) ^{2}-\sigma_2(I_1)) = a_{31},
  \]
  where \( I_1 = \{\alpha _{i_1},\cdots ,\alpha _{i_{k-1}}, \alpha _{i_{k-2}}\}\).
  Therefore, vectors in \(S _{k-1}\) are linearly dependent if and only if
  \[
    (a_{11}, a_{21}, a_{31}) ^{\T} = a_{11} \cdot (1, -\sigma_1(I_1), \sigma_1 ^{2}(I_1) - \sigma_2(I_1)) ^{\T}.
  \]
  In other words,  \(b_1 \in \langle b_I \rangle\), where \(b_I = (1, -\sigma_1(I_1), \sigma_1 ^{2}(I_1) - \sigma_2(I_1)) ^{\T}\) is a vector in \(\Omega _{k-2}\). In special, if \( a_{11} = 0\), this condition fails automatically. Therefore, condition (i) of \zcref{lem:NMDS} is equivalent to the theorem condition (3).

  \textbf{Condition (ii)}: Now, we show there exists a set of \(k\) columns is linearly dependent. We consider the set \(S _{k}\). If \( t _{k} = 0\) or \( t _{k} =3\), then the determinant of the matrix formed by vectors in \(S_k\) is nonzero.

  If \(t _{k} = 2\), without loss of generality, we assume \( S _{k} = \{g _{i_1},g _{i_2},\cdots , g _{i _{k-2}}, g_{n+1}, g_{n+2}\}\), and \( I_2 = \{\alpha _{i_1},\cdots ,\alpha _{i _{k-2}}\}\). By Laplace expansion, we expand along the last two columns as follows:
  \begin{align*}
    D_{k} & = \begin{vmatrix}
                v_{i_1}                   & v_{i_2}                   & \cdots & v_{i_{k-2}}                       & 0      & 0      \\
                v_{i_1}\alpha_{i_1}       & v_{i_2}\alpha_{i_2}       & \cdots & v_{i_{k-2}}\alpha_{i_{k-2}}       & 0      & 0      \\
                \vdots                    & \vdots                    & \ddots & \vdots                            & \vdots & \vdots \\
                v_{i_1}\alpha_{i_1}^{k-3} & v_{i_2}\alpha_{i_2}^{k-3} & \cdots & v_{i_{k-2}}\alpha_{i_{k-2}}^{k-3} & a_{11} & a_{12} \\
                v_{i_1}\alpha_{i_1}^{k-2} & v_{i_2}\alpha_{i_2}^{k-2} & \cdots & v_{i_{k-2}}\alpha_{i_{k-2}}^{k-2} & a_{21} & a_{22} \\
                v_{i_1}\alpha_{i_1}^{k-1} & v_{i_2}\alpha_{i_2}^{k-1} & \cdots & v_{i_{k-2}}\alpha_{i_{k-2}}^{k-1} & a_{31} & a_{32}
              \end{vmatrix} \\[4mm]
          & = \begin{vmatrix}
                a_{11} & a_{12} \\
                a_{21} & a_{22}
              \end{vmatrix} \cdot C_1 \;+\;
    \begin{vmatrix}
      a_{11} & a_{12} \\
      a_{31} & a_{32}
    \end{vmatrix} \cdot C_2 \;+\;
    \begin{vmatrix}
      a_{21} & a_{22} \\
      a_{31} & a_{32}
    \end{vmatrix} \cdot C_3,
  \end{align*}
  where
  \begin{align*}
    C_1 & = \bigl(\sigma_1^{2}(I_2) - \sigma_2(I_2)\bigr) \prod_{1\leq s<t\leq k-2} (\alpha_{i_t} - \alpha_{i_s}) \cdot \prod_{l=1}^{k-2} v_{i_l}, \\[2mm]
    C_2 & = \sigma_1(I_2) \prod_{1\leq s<t\leq k-2} (\alpha_{i_t} - \alpha_{i_s}) \cdot \prod_{l=1}^{k-2} v_{i_l},                                 \\[2mm]
    C_3 & = \prod_{1\leq s<t\leq k-2} (\alpha_{i_t} - \alpha_{i_s}) \cdot \prod_{l=1}^{k-2} v_{i_l}.
  \end{align*}
  Let \(b_I = (1, -\sigma_1(I_2), \sigma_1 ^{2}(I_2)- \sigma_2(I_2)) ^{\T}\). Then we have \(D_k\neq 0\) if and only if \(\det(b_1,b_2,b_I)\neq 0\).

  If \( t_k =1\), we assume \( g _{n+1}\) is contained in \(S_k\). We expand along the last columns by Laplace expansion as follows:

  \begin{align*}
    D_k & = \begin{vmatrix}
              v_{i_1}                   & v_{i_2}                   & \cdots & v_{i_{k-1}}                       & 0      \\
              v_{i_1}\alpha_{i_1}       & v_{i_2}\alpha_{i_2}       & \cdots & v_{i_{k-1}}\alpha_{i_{k-1}}       & 0      \\
              \vdots                    & \vdots                    & \ddots & \vdots                            & \vdots \\
              v_{i_1}\alpha_{i_1}^{k-3} & v_{i_2}\alpha_{i_2}^{k-3} & \cdots & v_{i_{k-1}}\alpha_{i_{k-1}}^{k-3} & a_{11} \\
              v_{i_1}\alpha_{i_1}^{k-2} & v_{i_2}\alpha_{i_2}^{k-2} & \cdots & v_{i_{k-1}}\alpha_{i_{k-1}}^{k-2} & a_{21} \\
              v_{i_1}\alpha_{i_1}^{k-1} & v_{i_2}\alpha_{i_2}^{k-1} & \cdots & v_{i_{k-1}}\alpha_{i_{k-1}}^{k-1} & a_{31}
            \end{vmatrix} \\[4mm]
        & = a_{11} C'_1 + a_{21} C'_2 + a_{31} C'_3                                                                                                                     \\[2mm]
        & = \bigl( b_{J_1}^{\mathsf{T}} \cdot b_1 \bigr) \cdot \prod_{1 \leq s < t \leq k-1} (\alpha_{i_t} - \alpha_{i_s}) \cdot \prod_{l=1}^{k-1} v_{i_l},
  \end{align*}
  where \(b_{J_1} = \bigl( \sigma_2(J_1),\; \sigma_1(J_1),\; 1 \bigr)^{\mathsf{T}}\) with \(J_1 = \{\alpha_{i_1}, \alpha_{i_2}, \dots, \alpha_{i_{k-1}}\}\). Hence, \(D_k = 0\) if and only if \(b_{J_1}^{\mathsf{T}} \cdot b_1 = 0\).

  \textbf{Condition (iii)}:
  To see this, let \(S _{k+1}\) be a set of \(k+1\) columns of \(G _{\GRL}\). If \( t _{k+1} = 0\) or \( t _{k+1} = 1\), then the vectors in \(S _{k+1}\) not lying in \( \{g_{n+1}, g_{n+2}, g_{n+3}\}\) from a Vandermonde matrix. If \( t_{k+1} =3\), as the discussion of the condition (i), then \( S _{k+1}\) contains a \( k\)-subset with \( \{g_{n+1}, g_{n+2}, g_{n+3}\}\) which is linearly independent.

  If \( t _{k+1} = 2\), we assume \( S _{k+1} = \{g _{i_1},g _{i_2},\cdots ,g _{i _{k-1}}, g_{n+1}, g_{n+2}\}\). Since \( g _{i_1}, g_{i_2},\cdots ,g_{i _{k-2}}\)  are linearly independent, the matrix formed by \( S _{k+1}\) has rank \( k-1\) if and only if \(  g _{n+1}\) and \( g _{n+2}\) belong to \( \langle g _{i_1},\cdots , g_{ i _{k-2}}\rangle\). By the discussion of condition (ii), this is equivalent to \(b_1\) and \(b_2\) are orthogonal to the vector \( b _{J_2} = (\sigma _{2}(J_2), \sigma _{1}(J_2), 1)^{\T}\) with \( J_2 = \{\alpha _{i_1},\cdots ,\alpha _{i _{k-1}}\}\). Therefore, every \( (k+1)\)-subset of columns of \( G  _{\GRL}\) has rank \( k\) if and only if, for every \( b_J \in \Gamma _{k-1}(\ba)\), at most one of \(b_1,b_2,b_3\) is orthogonal to \(b_J\). This is exactly the condition (2) of the theorem.

\end{proof}

\section{ Hermitian self-orthogonal GRL codes}

In this section, we move from NMDS characterization to Hermitian self-orthogonal constructions. We construct four families of Hermitian self-orthogonal GRL codes. In particular, for \(s=2\), our Hermitian self-orthogonal GRL codes are also NMDS codes.

To organize the discussion, we first present natural criteria for GRL codes with \(s=2\) and \(s=3\) to be Hermitian self-orthogonal in the following two lemmas.

For an element \(x \in \FQQ\), let \(\overline{x} = x ^{q}\). For \(\bv = (v_1,\cdots ,v_n)\in (\FQQ ^{*}) ^{n}\) and \(n\) distinct elements \(\ba = (\alpha_1,\cdots ,\alpha _{n})\), we define
\[
  S_t^{(H)}:=\sum_{i=1}^n v_i\,\bar v_i\,\alpha_i^t
  =\sum_{i=1}^n N(v_i)\alpha_i^t,
  \quad t\ge 0.
\]
Given a matrix \(A\), we denote \(\overline{A}\) as the matrix obtained by taking the conjugate of each entry of \(A\).

\begin{lemma}\label{thm:SO2}
  Let \(A_2\) is a nonsingular \(2 \times  2\) matrix over \(\FQQ\), \( \ba\) an \( n\)-subset of \( \FQQ\), \( \bv \in (\FQQ ^{*}) ^{n}\). Then \(\GRL _{k}(\ba,\bv, A_2)\) is Hermitian self-orthogonal if and only if the following condition holds:
  \begin{enumerate}
    \item If \(k>2\), \(S _{r+qs} ^{(H)}=0\) for all \(0\leq r \leq k-3, 0\leq s \leq k-1\);
    \item \(
          A_2\,\overline{A_2}^{\!\top}=
          -\begin{pmatrix}
            S_{(k-2)(1+q)}^{(H)}           & S_{k-2+q k-q}^{(H)}  \\[2mm]
            \overline{S_{k-2+q k-q}^{(H)}} & S_{(k-1)(1+q)}^{(H)}
          \end{pmatrix}.
          \)
  \end{enumerate}
\end{lemma}
\begin{proof}
  Let \(\{h_0,h_1,\cdots ,h _{k-1}\}\) be the rows of the generator matrix of \( \GRL _{k}(\ba,\bv, A_2)\). For \(0\leq  r,s \leq k-3\), the Hermitian inner product of \(h_r\) and \(h_s\) is
  \[
    \langle h_r, h_s \rangle _{H} = \sum _{i=1} ^{n} v_i \alpha _{i} ^{r} \overline{v_i \alpha _{i} ^{s}} = \sum _{i=1} ^{n} N(v_i)  \alpha _{i} ^{r+qs} = S _{r+qs} ^{(H)}.
  \]
  Moreover, as the last two columns of \(h_r\) are zero, when \(0\leq  r\leq k-3\), we have
  \[
    \langle h_r, h _{k-2} \rangle _{H} = S _{r+q(k-2)} ^{(H)}, \qquad \langle h_r, h _{k-1} \rangle _{H} = S _{r+q(k-1)} ^{(H)}.
  \]
  Assume \(A_2 = \begin{pmatrix}
    a & b \\
    c & d
  \end{pmatrix}\), then \(h _{k-2} = (\cdots , a,b)\) and \(h _{k-1}= (\cdots ,c,d)\). It follows that
  \begin{align*}
    \langle h _{k-1}, h _{k-1} \rangle _{H} & = S _{(q+1)(k-2)} ^{(H)} + a ^{q+1} + b ^{q+1},             \\
    \langle h _{k-2}, h _{k-2} \rangle _{H} & = S _{(q+1)(k-1)} ^{(H)} + c ^{q+1} + d ^{q+1},             \\
    \langle h _{k-2}, h _{k-1} \rangle _{H} & = S _{k-2+q(k-1)} ^{(H)} + a \overline{c} + b \overline{d}.
  \end{align*}
  Since the code \(\GRL _{k}(\ba,\bv, A_2)\) is Hermitian self-orthogonal if and only if for any \(0\leq  r,s \leq  k-3\), and \(t_1,t_2=0,1\),
  \[
    \langle h_r, h_s \rangle _{H} = \langle h_r, h _{k-2+t_1} \rangle _{H}= \langle h _{k-2+t_1}, h _{k-2+t_2} \rangle _{H} = 0,
  \]
  this is equivalent to the condition in the theorem.
\end{proof}
\begin{rk}
  \begin{enumerate}
    \item The second condition in \zcref{thm:SO2} implies that
          \[
            \overline{S _{(k-2)(1+q)} ^{(H)}} = S _{(k-2)(1+q)} ^{(H)}, \qquad \overline{S _{(k-1)(1+q)} ^{(H)}} = S _{(k-1)(1+q)} ^{(H)}.
          \]
          But by the definition of \(S_t ^{(H)}\), this is always true.
    \item In \zcref{thm:SO2}, if \(k=2\), the first condition is needless. \(\GRL _{2}(\ba,\bv, A_2)\) is Hermitian self-orthogonal if and only if the second condition holds.
    \item By the parity check matrix of GRL codes (cf. \cite{Li_2025}), we can also get an equivalent condition for the Hermitian self-orthogonality of GRL codes in terms of polynomials.

  \end{enumerate}
\end{rk}

\begin{cor}\label{cor:C4}
  For \(\bv = (v_1,\cdots ,v_n)\in (\FQQ ^{*})^{n}\) and an \(n\)-subset \(\ba = \{\alpha _{1}, \cdots ,\alpha _{n}\}\) of \( \Fq\). Assume there exists \(\lambda \in \FQQ ^{*}\) with \(v_i ^{1+q} = \lambda w_i\) for all \(1\leq  i\leq n\).
  For an integer \( k\leq  n\), we define
  \[
    \mu_1:=\sum_{i=1}^n w_i\,\alpha_i^{(k-2)(1+q)},\quad
    \mu_2:=\sum_{i=1}^n w_i\,\alpha_i^{k-2+q (k-1)},\quad
    \mu_3:=\sum_{i=1}^n w_i\,\alpha_i^{(k-1)(1+q)}.
  \]
  If \(n =  (k-1)(q+1)\) , and the nonsingular matrix \(A_2\) satisfies
  \[
    A_2\,\overline{A_2}^{\!\top}
    =-\lambda
    \begin{pmatrix}
      \mu_1     & \mu_2 \\[2mm]
      \bar\mu_2 & \mu_3
    \end{pmatrix},
  \]
  then \(\GRL _{k}(\bv,\ba,A_{2})\) is Hermitian self-orthogonal.
\end{cor}
\begin{proof}
  Since \(v_i ^{1+q} = \lambda w_i\), we have
  \[
    S_t ^{(H)} = \lambda \cdot \sum _{i=1} ^{n} w_i \alpha _{i} ^{t} \text{ for every } t\geq 0.
  \]
  By \zcref{solution}, we have \(S _{r+qs}^{(H)}=0\) for all \(0\leq  r\leq k-3\), and \(0\leq s\leq k-1\). In other words, the first condition in \zcref{thm:SO2} holds. Moreover, we have
  \begin{align*}
    S _{(k-2)(1+q)} ^{(H)} & = \lambda \cdot \left(\sum _{i=0}^{n} w_i \alpha _{i} ^{(k-2)(1+q)}\right)=\lambda \cdot\mu_1,    \\
    S _{k-2+q(k-1)} ^{(H)} & = \lambda \cdot \left(\sum _{i=0}^{n} w_i \alpha _{i} ^{k-2+q(k-1)}\right) = \lambda \cdot \mu_2, \\
    S _{(k-1)(1+q)} ^{(H)} & = \lambda \cdot \left(\sum _{i=0}^{n} w_i \alpha _{i} ^{(k-1)(1+q)}\right)=\lambda \cdot\mu_3.
  \end{align*}
  Then, the second condition in \zcref{thm:SO2} also holds. Hence, \(\GRL _{k}(\bv,\ba,A_{2})\) is Hermitian self-orthogonal.

\end{proof}
\begin{rk}
  In the above corollary, if \(n\geq (k-1)(q+1)+1\), then \(A_2\) must be a Hermitian self-orthogonal matrix. In other words, \(A_2 \overline{A_2}^{\T}=0\). This is a contradiction with the fact that \(A_2\) is a non-singular matrix. Thus, the condition \(n \leq  (k-1)(q+1) \) is necessary. But \(n\geq (k-1)(q+1)\) is not sufficient condition. Please see the next example for details.
\end{rk}

\begin{ex}
  Let \(q =3\), \(\FQQ  = \Fq(\omega)\) with \( \omega^{2} =-1\). We set
  \begin{align*}
    \ba & = \{\omega ,2\omega, 1, 2+\omega, 2+2\omega\}, \\
    \bv & = (1, 1, 1+\omega, 1,1),
  \end{align*}
  and
  \[
    A_2 = \begin{pmatrix}
      0        & \omega  \\
      1+\omega & 2\omega
    \end{pmatrix}.
  \]
  Then the GRL code \(\GRL _{3}(\ba,\bv, A_2)\) is a Hermitian self-orthogonal code with parameters \([7,3,4]\).

\end{ex}
\begin{proof}
  Since \(\bv ^{1+q}= (1,1,2,1,1)=2 \cdot (w_1,w_2,w_3,w_4,w_5)\), and,
  \[
    A_2 \overline{A_2} ^{\T} = \begin{pmatrix}
      1 & 2 \\
      2 & 0
    \end{pmatrix},
  \]
  \(\GRL _{3}(\ba,\bv,A_2)\) is Hermitian self-orthogonal if and only if
  \begin{align*}
    S_0 ^{(H)} = S_3 ^{(H)} = S_6 ^{(H)} = 0, \\
    S_4 ^{(H)} = 2, \quad S_7 ^{(H)} = 1, \quad S_8 ^{(H)} = 0.
  \end{align*}
  It is easy to check that the above condition holds. Moreover, the minimum distance of \(\GRL _{3}(\ba,\bv,A_2)\) is 4, which follows from the fact
  \[
    2\omega = \omega [2\omega +(2+\omega)],
  \]
  and \zcref{NMDS2}.

\end{proof}

\begin{lemma}\label{thm:SO3}
  Let \(A_3\) be a nonsingular \(3 \times  3\) times matrix over \(\FQQ\). Then \(\GRL _{k}(\ba,\bv, A_3)\) is Hermitian self-orthogonal if and only if the following condition holds:
  \begin{enumerate}
    \item If \(k>3\), \(S _{r+qs} ^{(H)}=0\) for all \(0\leq r\leq k-4\), \(0\leq s\leq k-1\);
    \item \[
            A_3\,\overline{A_3}^{\!\top}
            =-
            \begin{pmatrix}
              S_{(k-3)(1+q)}^{(H)}               & S_{(k-3)(1+q)+q}^{(H)}            & S_{(k-3)(1+q)+2q}^{(H)} \\[2mm]
              \overline{S_{(k-3)(1+q)+q}^{(H)}}  & S_{(k-2)(1+q)}^{(H)}              & S_{(k-2)(1+q)+q}^{(H)}  \\[2mm]
              \overline{S_{(k-3)(1+q)+2q}^{(H)}} & \overline{S_{(k-2)(1+q)+q}^{(H)}} & S_{(k-1)(1+q)}^{(H)}
            \end{pmatrix}.
          \]
  \end{enumerate}
\end{lemma}

\begin{proof}
  Let \( G_{\GRL}\) be the generator matrix of \( \GRL _{k}(\ba,\bv, A_3)\), \( \{h_1,\cdots ,h _{k}\}\) be the rows of \( G _{\GRL}\). Then \( \GRL _{k}(\ba, \bv, A_3)\) is a Hermitian self-orthogonal code if and only if  \( G _{\GRL} \overline{G _{\GRL}}^{\T} = 0\). In other words, for any \( 1\leq  r\leq k-4\), and \( 1\leq s \leq k-1\), we have
  \[
    \langle h_r, h_s \rangle _{H} = S _{r+qs} ^{(H)} = 0,
  \]
  and for any \( t_1, t_2 \in \{1,2\}\),
  \[
    \langle h _{k-3}, h _{k-3+t_1}\rangle _{H} = \langle h _{k-3+t_1}, h_{k-3+t_2} \rangle _{H} = 0.
  \]
  Let \( A_3 = \begin{pmatrix} a_{11} & a_{12} & a_{13} \\ a_{21} & a_{22} & a_{23} \\ a_{31} & a_{32} & a_{33} \end{pmatrix} \). Then the above six equations can be written as
  \begin{align*}
    \langle h _{k-3}, h_{k-3} \rangle _{H} =  & S _{(k-3)(1+q)} ^{(H)} + a_{11} ^{q+1} + a_{12} ^{q+1} + a_{13} ^{q+1},                                     \\
    \langle h _{k-3}, h _{k-2} \rangle _{H} = & S _{(k-3)(1+q)+q} ^{(H)} + a_{11} \overline{a_{21}} + a_{12} \overline{a_{22}} + a_{13} \overline{a_{23}},  \\
    \langle h _{k-3}, h _{k-1} \rangle _{H} = & S _{(k-3)(1+q)+2q} ^{(H)} + a_{11} \overline{a_{31}} + a_{12} \overline{a_{32}} + a_{13} \overline{a_{33}}, \\
    \langle h _{k-2}, h _{k-2} \rangle _{H} = & S _{(k-2)(1+q)} ^{(H)} + a_{21} ^{q+1} + a_{22} ^{q+1} + a_{23} ^{q+1},                                     \\
    \langle h _{k-2}, h _{k-1} \rangle _{H} = & S _{(k-2)(1+q)+q} ^{(H)} + a_{21} \overline{a_{31}} + a_{22} \overline{a_{32}} + a_{23} \overline{a_{33}},  \\
    \langle h _{k-1}, h _{k-1} \rangle _{H} = & S _{(k-1)(1+q)} ^{(H)} + a_{31} ^{q+1} + a_{32} ^{q+1} + a_{33} ^{q+1}.
  \end{align*}
  The condition (ii) in the theorem is equivalent to the above six equations hold. Hence, the code \(\GRL _{k}(\ba,\bv, A_3)\) is Hermitian self-orthogonal if and only if the two conditions in the theorem hold.
\end{proof}

\begin{cor}
  Let \( n,k\) be integers with  \(n = (k-1)(q+1)-1    \), and suppose there exists \(\lambda \in \FQQ ^{*}\) such that
  \[
    v_i ^{1+q} = \lambda w_i \text{ for all } 1\leq i \leq n.
  \]
  For \(0\leq  j, \ell \leq  2\), let
  \[
    \mu _{j\ell}:= \sum _{i=1} ^{n} w_i \alpha _{i} ^{(k-3 +j)+q(k-3 +\ell)}.
  \]
  If \[
    A_3 \overline{A_3} ^{\T} = - \lambda \cdot
    \begin{pmatrix}
      \mu _{00}     & \mu _{01}     & \mu _{02} \\[2mm]
      \bar\mu _{01} & \mu _{11}     & \mu _{12} \\[2mm]
      \bar\mu _{02} & \bar\mu _{12} & \mu _{22}
    \end{pmatrix},
  \]
  then \(\GRL _{k}(\ba,\bv, A_3)\) is Hermitian self-orthogonal.
\end{cor}
\begin{proof}
  The proof of this corollary is similar to that of \zcref{cor:C4}. We omit the details.
\end{proof}

\begin{ex}
  Let \(q=3\), \(\FQQ = \Fq(\omega)\) with \( w^{2} =-1\). We set
  \(
  \ba = \{\omega, 1+2\omega,2\},
  \bv = (1, 1+\omega, \omega),
  \)
  and
  \[
    A_3 = \begin{pmatrix}
      2+\omega  & 1+\omega  & 1        \\
      2+2\omega & 1+2\omega & 2+\omega \\
      2\omega   & 0         & 1
    \end{pmatrix}.
  \]
  Then the GRL code \(\GRL _{3}(\ba,\bv,A_3)\) is a Hermitian self-orthogonal code with parameters \([6,3,3]\).
\end{ex}
\begin{proof}
  As \(k=3\), we just need to check the second condition in \zcref{thm:SO3}. Moreover, since
  \[
    A_3 \overline{A_3} ^{\T} = \begin{pmatrix}
      2        & 2\omega+2 & 2\omega \\
      \omega+2 & 0         & 0       \\
      \omega   & 0         & 2       \\
    \end{pmatrix},
  \]
  the Hermitian self-orthogonal property follows from the fact:
  \[
    S_0 ^{(H)} = 1, \quad S_3 ^{(H)} = \omega+1, \quad S_6 ^{(H)} = \omega,\quad S_4 ^{(H)} = 0, \quad S_7 ^{(H)} = 0, \quad S_8 ^{(H)} = 1.
  \]
  For \(\ba=\{\omega,1+2\omega,2\}\), we have
  \begin{align*}
    \Omega_1 = \left\{ \begin{pmatrix} 1 \\ \omega \\ 2 \end{pmatrix}, \begin{pmatrix} 1 \\ 1 + 2\omega \\ \omega \end{pmatrix}, \begin{pmatrix} 1 \\ 2 \\ 1 \end{pmatrix} \right\}, \\
    \Gamma_2 = \left\{ \begin{pmatrix} 1 + \omega \\ 2 \\ 1 \end{pmatrix}, \begin{pmatrix} 2\omega \\ 1 + 2\omega \\ 1 \end{pmatrix}, \begin{pmatrix} 2 + \omega \\ \omega \\ 1 \end{pmatrix} \right\}.
  \end{align*}
  Let \(b_1,b_2,b_3\) be the columns of \(A_3\). Then we have
  \[
    \det\!\left(b_1,b_2,\begin{pmatrix}1\\2\\1\end{pmatrix}\right)=0.
  \]
  Moreover, for the three vectors in \(\Gamma _{2}\), we have
  \[
    \begin{array}{c|ccc}
      b_J                        & b_1^\top b_J & b_2^\top b_J & b_3^\top b_J \\ \midrule
      (1+\omega,2,1)^\top        & 2            & 2            & 0            \\
      (2\omega,1+2\omega,1)^\top & 2            & 1            & 1+\omega     \\
      (2+\omega,\omega,1)^\top   & 1+2\omega    & 2+\omega     & 2
    \end{array}
  \]
  Hence, each vector \(b_J\) is orthogonal to at most one columns of \(A_3\). On the other hand, it is not hard to check for any \(b_I \in \Omega_1\), \(b_j \notin \langle b_I \rangle\). By \zcref{NMDS3}, the minimum distance of \(\GRL _{3}(\ba,\bv,A_3)\) is 3.

\end{proof}
\subsection{The construction of Hermitian self-orthogonal GRL codes}
In this subsection, we apply \zcref{NMDS2}, \zcref{thm:SO2}, and \zcref{thm:SO3} to derive explicit constructions of Hermitian self-orthogonal GRL codes. The first two families are NMDS codes.

We first study the construction of Hermitian self-orthogonal GRL codes with \(s=2\), and then extend the approach to \(s=3\).

\begin{thm}\label{thm:SO2NMDS}
  Let \(q\) be a prime power with \(q \geq  4\), and \(m,k\) be positive integers with \(2\leq  m\leq q-2\), and \(3\leq k\leq m+1\).  Then there exists a \(q ^{2}\)-ary Hermitian self-orthogonal GRL code with parameters
  \[
    [m(q+1)+2, k, m(q+1)+2-k]_{q ^{2}}.
  \]
\end{thm}
\begin{proof}
  Let \(n = m(q+1)\), and \(I=\{c_1,c_2,\cdots ,c_m\}\) be an \(m\)-subset of \(\Fq^*\) such that \(\sigma:=\sum _{j=1}^{m}c_j \neq 0\). Actually, we always can find a such subset \(I\) which has nonzero sum since \(m\leq q-2\).

  We define
  \(\beta _{j}\in \FQQ ^{*}\) by \(\beta _{j} ^{q+1} = c_j\) for \(1\leq j \leq m\). And let \(U _{q+1}\) be the subgroup of \(\FQQ ^{*}\) with order \(q+1\). We choose \(n\) distinct elements in \(\FQQ\) as follows:
  \[
    \ba := \bigcup _{j=1} ^{m} \beta _{j} \cdot U _{q+1} \triangleq \bigcup _{j=1} ^{m} C_j.
  \]
  It is not hard to check that for any \(1\leq i\neq j \leq m\), \(C_i \cap C_j =\varnothing\). For any \(\alpha \in C_j\), let \(v_{\alpha}\in \FQQ^{*}\) satisfying
  \[
    v _{\alpha} ^{q+1} = \frac{c_j ^{m-k+1}}{D_j} \in \Fq ^{*},
  \]
  where \(D_j =\prod _{ \ell =1,\ell \neq  j} ^{m} (c_j -c _{\ell})\). We will denoted \(\bv\) as the vector with coordinate \(v_{\alpha}\) for \(\alpha \in \ba\). Moreover, let \(\xi \) be a primitive \((q+1)\)-th root of unity in \(\FQQ\), and
  \[
    \theta _{j} :=\beta _{j} \cdot \frac{1-\xi ^{k-1}}{1-\xi}, \quad 1\leq j \leq m.
  \]
  We choose \(j_0\) such that \( \sigma + (\theta _{j_0} ^{q+1})\neq 0\), and define
  \[
    \rho := - \frac{\sigma}{\sigma +\theta _{j_0} ^{q+1}} \in \Fq ^{*}.
  \]
  Then we choose \(a,d \in \FQQ ^{*}\) such that
  \[
    a ^{q+1}  = \rho ,\qquad d ^{q+1} = - \frac{\sigma ^{2}}{\sigma+\theta _{j_0} ^{q+1}}.
  \]
  Keep these notation in mind, we define
  \[
    A_2 = \begin{pmatrix}
      a                    & - \rho \cdot \overline{\left(\frac{\theta _{j_0}}{d}\right)} \\
      \theta _{j_0}\cdot a & d
    \end{pmatrix}.
  \]
  We claim that the GRL code \(\GRL _{k}(\ba,\bv,A_2)\) is a Hermitian self-orthogonal NMDS code.

  We first check the Hermitian self-orthogonal property. For \(0\leq r \leq k-3\), and \(0\leq  s\leq k-1\),
  \[
    S _{r+qs} ^{(H)} = \sum _{\alpha \in \ba} v _{\alpha} ^{q+1} \alpha ^{r+qs} = \sum _{j=1} ^{m} \sum _{\alpha\in C_j} \frac{c_j ^{m-k+1}}{D_j} \cdot \alpha ^{r+qs} = \sum _{j=1}^{m}  \frac{c_j^{m-k+1+s}}{D_j} \cdot \sum _{\alpha\in C_j}\alpha ^{r-s}.
  \]
  The last equation follows from the fact that \(\alpha ^{q+1} = c_j\) for \(\alpha \in C_j\). Since \(C_j\) is a coset of \(U _{q+1}\), for \(t \in \mathbb{N}\), we have
  \[
    \sum _{\alpha \in C_j} \alpha ^{t} = \begin{cases}
      0,               & \text{ if } q+1 \nmid t, \\
      \beta _{j} ^{t}, & \text{ if } q+1 \mid t.
    \end{cases}
  \]
  Further, since \(|r-s|\leq k-1\leq m\leq q-2\), we have \(S _{r+qs} ^{(H)}=0\) for all \(r\neq s\).

  On the other hand, by \zcref{solution}, we have
  \[
    \sum_{j = 1}^{m}\frac{c_j^{\ell}}{D_j}=\begin{cases}
      0,      & \text{if }0\leq\ell\leq m - 2, \\
      1,      & \text{if }\ell=m - 1         , \\
      \sigma, & \text{if }\ell = m.
    \end{cases} \tag{$\star$}
  \]
  Thus, for \(0\leq r =s \leq  k-3\), we have
  \[
    S _{r+qs} ^{(H)} = \sum _{j=1} ^{m} \frac{c_j ^{m-k+1+r}}{D_j} = 0,
  \]
  since \(m-k+1+r \leq m-k+1+k-3 = m-2\). Hence, the first condition in \zcref{thm:SO2} holds.

  By the equation \((\star)\), we have
  \begin{align*}
    S _{(k-2)(1+q)} ^{(H)}   & = \sum _{ j=1} ^{m} \sum _{\alpha \in C_j} \frac{c_j ^{m-k+1}}{D_j} \cdot \alpha ^{(k-2)(1+q)} = (q+1)\cdot\sum _{j=1} ^{m} \frac{c_j ^{m-1}}{D_j} = 1,                                      \\
    S _{(k-2)(1+q)+q} ^{(H)} & = \sum _{j=1} ^{m} \sum _{\alpha \in C_j} \frac{c_j ^{m-k+1}}{D_j} \cdot c_j ^{k-2}\cdot \alpha ^{q} = \sum _{j=1} ^{m} \frac{c_j ^{m-1}}{D_j} \cdot \sum _{\alpha \in C_j} \alpha ^{q} = 0, \\
    S _{(k-1)(1+q)} ^{(H)}   & = \sum _{j=1} ^{m} \sum _{\alpha \in C_j} \frac{c_j ^{m-k+1}}{D_j} \alpha ^{(k-1)(1+q)} = \sum _{j=1} ^{m} \sum _{ \alpha \in C_j} \frac{c_j ^{m}}{D_j} = \sigma.
  \end{align*}
  Let
  \(
  M_2 = \begin{pmatrix}
    -1 & 0       \\
    0  & -\sigma
  \end{pmatrix}.
  \) Then the second condition in \zcref{thm:SO2} holds if and only if \(A_2 \overline{A_2} ^{\T} = M_2\). Actually, by the definition of \(A_2\),
  \[
    A_2 \overline{A_2}^{\T} = \begin{pmatrix}
      \rho + \rho ^{2}\cdot (\frac{\theta _{j_0}}{d}) ^{q+1}  & a ^{q+1}\cdot \overline{\theta _{j_0}}- \rho \cdot \overline{\theta _{j_0}} \\
      a ^{q+1} \cdot \theta _{j_0} - \rho \cdot \theta _{j_0} & \rho \cdot\theta _{j_0} ^{q+1} + d ^{q+1}
    \end{pmatrix}.
  \]
  Since
  \[
    \rho+ \rho ^{2}\cdot (\frac{\theta _{j_0}}{d}) ^{q+1} = -\frac{\sigma}{\sigma+\theta _{j_0} ^{q+1}} - \frac{\theta _{j_0}^{q+1}}{\sigma+\theta _{j_0} ^{q+1}} = -1,
  \]
  and
  \[
    \rho \cdot \theta _{j_0} ^{q+1} = - \frac{\sigma \cdot \theta _{j_0} ^{q+1}}{\sigma + \theta _{j_0} ^{q+1}}, \qquad d ^{q+1} = - \frac{\sigma ^{2}}{\sigma + \theta _{j_0} ^{q+1}},
  \]
  it implies that
  \(
  A_2 \overline{A_2} ^{\T} = M_2.
  \)

  Finally, since \(\theta _{j_0}\in \Delta _{k-1}\), \(\GRL _{k}(\ba,\bv,A_2)\) is a NMDS code with minimum distance \(m(q+1)+2-k\) by \zcref{NMDS2}. Hence, \(\GRL _{k}(\ba,\bv,A_2)\) is a Hermitian self-orthogonal NMDS code.

\end{proof}

\begin{rk}
  \begin{enumerate}
    \item Because \(M_2 =\overline{M_2}^{\T}\), \(M_2\) corresponds to a nonsingular Hermitian form. By \cite[Proposition 2.3.1]{Kleidman_1990}, there exists a nonsingular matrix \(P\) over \(\FQQ\) such that \(M_2 = P \overline{P} ^{\T}\). Hence, \(\GRL _{k}(\ba,\bv,P)\) is a Hermitian self-orthogonal code.
    \item 	The nonsingularity of \(A_2\) follows from the fact that \(A_2 \overline{A_2} ^{\T} = M_2\) is nonsingular.
  \end{enumerate}

\end{rk}

\begin{ex}
  Let \(q = 5\), and \(\FQQ = \Fq(\omega)\), where \(\omega ^{2}+2=0\). We set
  \begin{align*}
    \ba & = \left\{1, 2+\omega, 2-\omega, 3+\omega, 3-\omega, 4, \omega, 3+2\omega, 2+2\omega,     \right.        \\
        & \left. 3+3\omega, 2+3\omega, 4\omega, 2\omega, 1-\omega, 4-\omega, 1+\omega, 4+\omega, 3\omega\right\}, \\
    \bv & = \{ 2\omega, 2\omega, 2\omega, 2\omega, 2\omega, 2\omega, 1+2\omega, 1+2\omega, 1+2\omega,             \\
        & 1+2\omega, 1+2\omega, 1+2\omega, 2\omega, 2\omega, 2\omega ,2\omega,2\omega,2\omega\},
  \end{align*}
  and
  \[
    A_2 = \begin{pmatrix}
      1        & 4+\omega \\  \\
      1+\omega & 1
    \end{pmatrix}.
  \]
  Then \(\GRL _{4}(\ba,\bv,A_2)\) is a Hermitian self-orthogonal NMDS code with parameters \([20,4,16] _{25}\).

\end{ex}
\begin{proof}
  Let \(U_6 = \{1, 2+2\omega, 2-\omega, 3+\omega, 3-\omega, 4\}\). Then
  \[
    \ba = U_6 \ \cup \ \omega \cdot U_6 \ \cup \  2\omega \cdot U_6.
  \]
  Let \(I := (c_1,c_2,c_3) = (1,2,3)\) and \(\sigma :=c_1+ c_2+c_3 = 1\).
  Denote \(\xi = 3+\omega\). Then we have
  \[
    \theta_3 :=  \omega \cdot \frac{1-\xi^{3}}{1-\xi} = 2\omega \cdot (1+\xi +\xi^{2}) = 1+\omega.
  \]

  It is easy to see that \(\theta_3\) satisfies the condition \(\sigma +\theta_3 \neq 0\).
  Moreover, the following equalities are presented:
  \[
    \rho:= - \frac{\sigma}{\sigma+\theta_3^{q+1}} = 1, \quad -\frac{\sigma^{2}}{\sigma+ \theta_3^{q+1}} = 1, \quad - \rho \cdot \theta_3^{q} = 4+\omega.
  \]

  By the construction in the proof of \zcref{thm:SO2NMDS}, the code \(\GRL_{4}(\ba,\bv, A_2)\) is a Hermitian self-orthogonal NMDS code with parameters \([20,4,16] _{25}\).
  We also double check its parameters by Magma algebra system \cite{Magma}.
\end{proof}

\begin{thm}\label{thm2:SO2NMDS}
  Let \(q\) be a prime power with \(q\geq  5\), and \(m,k\) integers with \(2\leq m\leq q+1\), and \(3\leq k\leq \frac{q+1}{2}\). Then there exists a \(q ^{2}\)-ary Hermitian self-orthogonal GRL code with parameters
  \[
    [m(q-1)+2, k, m(q-1)+2-k] _{q ^{2}}.
  \]

\end{thm}
\begin{proof}
  Let \(n  = m(q-1)\), and \(R = \{\omega_1 = 1,\cdots ,\omega _{m}, \omega _{m+1},\cdots ,\omega _{q+1}\}\) be a representative set of \(\FQQ ^{*} /\Fq ^{*}\) in \(\FQQ ^{*}\). We choose \(n\) distinct elements in \(\FQQ\) as follows:
  \begin{align*}
    \ba & := \omega_1 \cdot \Fq ^{*} \ \cup \ \omega_2 \cdot \Fq ^{*} \ \cup \ \cdots \ \cup \ \omega_m \cdot \Fq ^{*} \\
        & \triangleq \left\{\alpha _{j,x} = \omega _{j} \cdot x: 1\leq j \leq m, x\in \Fq ^{*}\right\}.
  \end{align*}
  We choose \(u_1,u_2,\cdots ,u_m \in \Fq ^{*}\) (not necessarily distinct) such that
  \[
    \mu : = \sum _{j=1}^{m} u_j \omega _{j} ^{k-2+q(k-1)} \neq 0.
  \]
  We can always find such \(u_j\). If not, we arbitrarily choose \((u_j) _{j=1}^{m}\) and \(u_1 '\in \Fq ^{*}\) such that \(u_1'\neq u_1\), and \[
    \sum _{j=1}^{m} u_j \omega _{j} ^{k-2+q(k-1)} = \sum _{j=2}^{m} u_j \omega _{j} ^{k-2+q(k-1)} + u_1' \omega _{1} ^{k-2+q(k-1)} =0.
  \]
  It implies \(u_1 = u_1 '\), a contradiction. For each \(\alpha _{j,x}\), we choose \(v _{j,x} \in \FQQ ^{*}\) such that \(N(v _{j,x})=u_j \cdot x ^{E}\), where \(E = q-2k+2\). Then we define \(\bv\) as the vector with coordinate \(v _{j,x}\) for \(\alpha _{j,x}\in \ba\).

  Let \(f(y) = \mu y + \overline{\mu y}\) be a \(\Fq\)-linear map from \(\FQQ\) to \(\Fq\), and \(j_0\) an integer with \(1\leq j_0\leq m\) such that \(f(\omega _{j_0})\neq 0\). If there not exists such \(j_0\), then \(f(\omega_1) = f(\omega_2)=0\), which follows that \(\omega_1\) and \(\omega_2\) belong to the same coset of \(\Fq ^{*}\) in \(\FQQ ^{*}\), a contradiction. We set \(\theta:= \omega _{j_0}\), \(\tau: = f(\theta)\), and \(\rho = N(\mu)/\tau\). Choose \(a,d \in \FQQ ^{*}\) such that
  \[
    N(a) = \rho, \qquad N(d) = -\rho N(\theta).
  \]
  Then we define
  \[
    A_2  = \begin{pmatrix}
      a              & \frac{\mu - \rho \cdot \overline{\theta}}{\overline{d}} \\
      \theta \cdot a & d
    \end{pmatrix}.
  \]

  We claim that \(\GRL _{k}(\ba,\bv,A_2)\) is a Hermitian self-orthogonal NMDS code. For \(0\leq r\leq k-3\), and \(0\leq s\leq k-1\), we have
  \[
    S _{r+qs} ^{(H)} = \left( \sum _{j=1} ^{m} u_j \cdot \omega _{j} ^{r+qs}\right) \cdot \left( \sum _{x\in \Fq ^{*}} x ^{r+s+E}\right).
  \]
  As \(1\leq  r+s+E \leq q-2\), we have \(\sum _{x\in \Fq ^{*}} x ^{r+s+E} = 0\). Hence, the first condition in \zcref{thm:SO2} holds.

  Moreover,
  \begin{align*}
    S _{(k-2)(1+q)} ^{(H)}  & = \left(\sum _{j=1} ^{m} u_j \cdot \omega _{j} ^{(k-2)(1+q)}\right) \cdot \left( \sum _{x\in \Fq ^{*}}x ^{q-2}\right) = 0,    \\
    S _{(k-1)(1+q)} ^{(H)}  & = \left(\sum _{j=1} ^{m} u_j \cdot\omega _{j} ^{(k-1)(1+q)}\right) \cdot \left( \sum _{x \in \Fq ^{*}} x \right) = 0,         \\
    S _{k-2 +q(k-1)} ^{(H)} & = \left( \sum _{j=1} ^{m} u_j \cdot \omega _{j} ^{k-2+q(k-1)}\right) \cdot \left(\sum _{x\in \Fq ^{*}}x ^{q-1}\right) = -\mu.
  \end{align*}
  And
  \begin{align*}
    A_2 \overline{A_2} ^{\T} & = \begin{pmatrix}
                                   N(a) + \frac{N(\mu - \rho \cdot \overline{\theta})}{N(d)}         & N(a) \cdot \theta + \mu - \rho \cdot \overline{\theta} \\
                                   N(a) \cdot \overline{\theta} + \overline{\mu} - \rho \cdot \theta & N(\theta \cdot a) + N(d)
                                 \end{pmatrix}.
  \end{align*}
  Since \[N(\mu- \rho \cdot \overline{\theta}) = N(\mu) - (\rho \overline{\theta} \overline{\mu}+ \mu \overline{\rho}\theta)+N(\rho \overline{\theta}) = \tau \rho - \rho \cdot f(\theta) - \rho N(d),\]
  and
  \[
    N(a) \cdot \theta + \mu -\rho \cdot \overline{\theta} = \mu,
  \]
  we have
  \[
    A_2 \overline{A_2} ^{\T} = \begin{pmatrix}
      0              & \mu \\
      \overline{\mu} & 0
    \end{pmatrix}.
  \]
  Hence, the second condition in \zcref{thm:SO2} holds.

  Finally, by \zcref{lem:sum2} and \(k-1<q-1\), there exists \(k-1\) distinct elements \(x_1,\cdots ,x _{k-1} \in\Fq ^{*}\) such that \(\sum _{i=1} ^{k-1} x_i  = 1\). Hence, \(\theta = \omega _{j_0}\in \Delta _{k-1}\). By \zcref{NMDS2}, the minimum distance of \(\GRL _{k}(\ba,\bv,A_2)\) is \(m(q-1)+2-k\). Therefore, \(\GRL _{k}(\ba,\bv,A_2)\) is a Hermitian self-orthogonal NMDS code.

\end{proof}

\begin{ex}
  Let \(q=5\) and \(\FQQ = \Fq(z)\) with \(z ^{2}+2 = 0\). We set
  \[
    \ba :=\{1,2,3,4, z,2z,3z,4z\},\quad \bv :=(1,z,1+z,2,1,z,1+z,2),
  \]
  and
  \[
    A_2 = \begin{pmatrix}
      z & 1+4z \\
      z & 1+z
    \end{pmatrix}.
  \]
  Then the GRL code \(\GRL _{3}(\ba, \bv, A_2)\) is a Hermitian self-orthogonal NMDS code with parameters \([10,3,7] _{25}\).

\end{ex}
\begin{proof}
  Let \(m = 2\). We choose
  \[
    \omega_1 =1, \ \omega_2 = z, \ u_1=u_2 = 1.
  \]
  Then
  \[
    \mu = \omega_1 ^{11}+ \omega_2 ^{11} = 1+3 z \neq 0.
  \]
  Take \(\theta =\omega_1 =  1\), then \(\tau = 2\), and \(\rho = 2\). Let
  \[
    a = z, d=1+z.
  \]
  By the proof of the above theorem, the code \(\GRL _{3}(\ba,\bv,A_2)\) is a Hermitian self-orthogonal code with parameters \([10,3,7] _{25}\).

\end{proof}

The following table summarizes the known constructions of Hermitian self orthogonal NMDS codes over \( \FQQ\), including the two families of our constructions. It is worth noting that the parameters of the two families of Hermitian self-orthogonal NMDS codes constructed in this paper have flexible choices of length (always larger than \(q\)) and dimension which are not covered by the known constructions.
\begin{table}[H]
  \centering
  \small
  \renewcommand{\arraystretch}{1.1}
  \setlength{\tabcolsep}{5pt}
  \caption{Comparison of known constructions of Hermitian self-orthogonal NMDS codes}
  \label{tab:comparison-hermitian-self-orthogonal-nmds}

  \begin{tabularx}{\textwidth}{
    >{\centering\arraybackslash}m{1.8cm}
    >{\centering\arraybackslash}m{4.2cm}
    >{\centering\arraybackslash}m{2.4cm}
    >{\raggedright\arraybackslash}X}
    \toprule
    Field & Parameters                                                                                                                     & Reference & Construction method \\
    \midrule

    $\mathbb{F}_{9}$
          & \makecell[c]{$[10,5,5],$                                                                                                                                         \\ $[12,6,6]$}
          & \cite{2014NMDS}
          & Computer search                                                                                                                                                  \\
    \midrule

    $\mathbb{F}_{25}$
          & \makecell[c]{$[10,5,5],$                                                                                                                                         \\ $[12,6,6],$ \\ $[14,7,7]$}
          & \cite{2014NMDS}
          & Computer search                                                                                                                                                  \\
    \midrule

    $\mathbb{F}_{121}$
          & \makecell[c]{$[6,3,3],$                                                                                                                                          \\ $[8,4,4]$, \\ $[10,5,5]$}
          & \cite{2014NMDS}
          & Computer search                                                                                                                                                  \\
    \midrule

    $\mathbb{F}_{q^2}$
          & $[2k,k,k]$
          & \cite{Guo_2022}
          & Constructed from $(+)$-GTRS codes with $2\leq 2k \le q$                                                                                                          \\
    \midrule

    $\mathbb{F}_{q^2}$
          & $[2k,k,k]$
          & \cite{Zhu_2025}
          & Constructed from $(+)$-TGRS codes and $(*)$-TGRS codes with $2\leq 2k \le q$                                                                                     \\
    \midrule

    $\mathbb{F}_{q^2}$
          & \makecell[c]{$[m(q+1)+2,\,k,$                                                                                                                                    \\ $m(q+1)+2-k]$}
          & \zcref{thm:SO2NMDS}
          & Constructed from GRL codes with parameters $[m(q+1)+2,k]$, where $q \ge 4$, $2 \le m \le q-2$, and $3 \le k \le m+1$                                             \\
    \midrule

    $\mathbb{F}_{q^2}$
          & \makecell[c]{$[m(q-1)+2,\,k,$                                                                                                                                    \\ $m(q-1)+2-k]$}
          & \zcref{thm2:SO2NMDS}
          & Constructed from GRL codes with parameters $[m(q-1)+2,k]$, where $q \ge 5$, $2 \le m \le q+1$, and $3 \le k \le \frac{q+1}{2}$                                   \\
    \bottomrule
  \end{tabularx}
\end{table}

\subsection{s=3 cases}

Now, we turn to the more complex cases with \(s=3\). We first present a lemma which is crucial for the construction of Hermitian self-orthogonal GRL codes with \(s=3\).

\begin{lemma}\label{solution2}
  Let \(q\) be a prime power with \(q\geq 5\), and \(m\) an integer with \(2\leq  m\leq q-3\). Then there exists a \(m\)-subset \(C = \{c_1,c_2,\cdots ,c_m\}\subset \Fq ^{*}\) such that
  \begin{align*}
    \sigma_1(C) & = -\sum _{j=1} ^{m} c_j \neq 0,              \\
    \sigma_2(C) & = \sum _{1\leq  i< j \leq m} c_i c_j \neq 0.
  \end{align*}
\end{lemma}
\begin{proof}
  If \(m=2\), let \(\zeta\) be a primitive element of \(\Fq ^{*}\), and choose \(c_1=\zeta, c_2  = \zeta ^{2}\). Then we have
  \[
    \sigma_1(C)= - \zeta (1+\zeta)\neq 0, \quad \sigma_2(C) = \zeta ^{3} \neq 0.
  \]
  Assume we have found a valid \(i-1\) subset \(C _{i-1}\) for some \(3\leq  i \leq q-4\), which satisfies \(\sigma_1(C _{i-1})\neq 0\) and \(\sigma_2(C _{i-1})\neq 0\). Suppose \(x\in \Fq ^{*}\backslash C _{i-1}\) such that
  \[
    \sigma_1 - x = 0, \quad \text{or} \quad \sigma_2 - x \cdot \sigma_1  = 0,
  \]
  then at most two such \(x\) exits. Since \(q - (i-1) - 2 \geq  q-m-2\geq  1\), we can always find an element \(c_i \in \Fq ^{*}\) such that
  \[
    \sigma_1(C _{i-1}) - c_i \neq 0, \quad \text{and} \quad \sigma_2(C _{i-1}) - c_i \cdot \sigma_1(C _{i-1}) + c_i ^{2} \neq 0.
  \]
  Hence, we just let \(C_i = C _{i-1} \cup \{c_i\}\), and the lemma follows by induction.
\end{proof}

Now, following the construction methods in \zcref{thm:SO2NMDS} and \zcref{thm2:SO2NMDS}, we construct Hermitian self-orthogonal GRL codes for s=3 in \zcref{HSO_s3} and \zcref{HSO_s3_2}, respectively. But it cannot be determined whether these codes are NMDS or not, since the complexity conditions in \zcref{thm:SO3} are not easy to check for a family of codes. Relatively, we provided a lower bound of the minimum distance of these codes.

\begin{lemma}\label{Dist_s3}
  Let \(q\) be a prime power, \(A_3\) a nonsingular \(3 \times 3\) matrix over \(\FQQ\), \(\ba\) an \(n\)-subset of \(\FQQ\), and \(\bv\) a vector in \((\FQQ ^{*}) ^{n}\). Then the minimum distance of \(\GRL _{k}(\ba,\bv,A_3)\) is at least \(n-k+2\).
\end{lemma}
\begin{proof}
  Let \(c = (c_1,c_2)\) be a nonzero codeword in \(\GRL _{k}(\ba,\bv,A_3)\). We consider the following two cases:

  \textbf{Case 1}: If \(\wt(c_1) \geq n-k+2\), then \(\wt(c) \geq n-k+2\) immediately.

  \textbf{Case 2}: Now, suppose \(\wt(c_1) = n-k+1\). In this case, there exists a vector \(\mathbf{m} = (m_1,\dots,m_k) \in \FQQ ^{k}\) such that \(c = \mathbf{m} G_{\GRL}\).
  If \((m _{k-2}, m _{k-1}, m _{k}) = (0,0, 0)\), then \(c\) is a linear combination of the first \(k-2\) rows of the generator matrix \( G _{\GRL}\) of \( \GRL _{k}(\ba,\bv, A_3)\). By the definition of \(\GRL_k\), this implies that \(\wt(c_1) \geq n-(k-3)+1 = n-k+4\), which contradicts our assumption that \(\wt(c_1) = n-k+1\). Therefore, we must have \((m _{k-2},m _{k-1}, m _{k}) \neq (0,0, 0)\).
  Since \(A_3\) is nonsingular and \((m _{k-2}, m _{k-1}, m _{k}) \neq (0,0,0)\), the construction of the GRL codes implies that \(c_2\) is not the zero vector.
  This implies \(\wt(c_2) \geq 1\). Consequently, \(\wt(c) = \wt(c_1) + \wt(c_2) \geq (n-k+1) + 1 = n-k+2\).
\end{proof}

\begin{thm}\label{HSO_s3}
  Let \(q \geq 5\), \(2\leq m\leq  q-3\), and \(4\leq k\leq m+2\). Then there exists a \(q ^{2}\)-ary Hermitian self-orthogonal GRL code with parameters
  \[
    [   m(q+1)+3, k , \geq m(q+1)-k+2] _{q ^{2}}.
  \]

\end{thm}
\begin{proof}
  Let \(C_m = \{c_1,c_2,\cdots ,c_m\} \subset \Fq ^{*}\) be an \(m\)-subset and define
  \[
    D_j :=\prod _{\ell =1, \ell \neq j} ^{m} c_j - c_{\ell},\quad  \Sigma:=\sum _{j=1} ^{m}c_j =  - \sigma_1(C_m),\quad  \Lambda = \sum _{j=1} ^{m} \frac{c_j ^{m+1}}{D_j}, \quad 1\leq j \leq m.
  \]
  We claim that there exists \(C_m\) such that
  \(
  \Sigma \neq 0,  \Lambda \neq 0.
  \)
  By \zcref{solution}, we have
  \[
    \Lambda = \sum _{j=1} ^{m} c_j ^{2} + \sum _{1\leq i < j \leq m} c_i c_j = \sigma_1(C_m) ^{2} - \sigma_2(C_m).
  \]
  Let \(A_m: = \Fq ^{*} \backslash C_m\). As  \(\sum _{x \in \Fq ^{*}} x = 0\), we have
  \[
    \Sigma = - \sigma_1(C_m) \neq 0 \iff \sigma_1(A_m) \neq 0.
  \]
  At the same time, consider the coefficient of \(y ^{q-3}\) in the following two polynomials:
  \[
    \prod _{x_a \in A_m}(y-x_a)\prod _{x_c \in C_m} (y-x_c) = y ^{q-1} - 1,
  \]
  which implies
  \[
    \sigma_2(C_m) = \sigma_1(A_m) ^{2} - \sigma_2(A_m).
  \]
  Then we have \[
    \Lambda \neq 0 \iff \sigma_2(A_m) \neq 0 .
  \]
  Therefore, to prove the claim, we just need to find a \((q-1-m)\)-subset \(A_m\) such that \(\sigma_1(A_m)\neq 0\) and \(\sigma_2(A_m)\neq 0\). By \zcref{solution2}, such \(A_m\) always exists.

  For each \(j\), let
  \[
    \mathcal{C}_j = \{ \alpha \in \FQQ ^{*} : \alpha ^{q+1} = c_j\},
  \]
  and
  \(
  \ba := \bigcup _{j=1} ^{m} \mathcal{C}_j.
  \) Thus \(|\ba| = m(q+1)\). For any \(\alpha \in \mathcal{C} _{j}\), we choose \(v_{\alpha} \in \FQQ ^{*}\) such that \(N(v_\alpha) = \frac{c_j ^{m-k+2}}{D_j}\), and choose \(a,b,c \in \FQQ ^{*}\) with \(N(a) = -1\), \(N(b) = - \Sigma\), and \(N(c) = - \Lambda\). We define
  \[
    A_3 = \begin{pmatrix}
      a & 0 & 0 \\
      0 & b & 0 \\
      0 & 0 & c
    \end{pmatrix}.
  \]
  Then we claim that the GRL code \(\GRL _{k}(\ba,\bv, A_3)\) is a Hermitian self-orthogonal code, where \(\bv\) is the vector with coordinate \(v_{\alpha}\) for \(\alpha \in \ba\).

  To show Hermitian self-orthogonal property, we first check the first condition in \zcref{thm:SO3}. For \(0\leq r\leq k-4\), and \(0\leq s \leq k-1\),
  \[
    S _{r+qs} ^{(H)} = \sum _{j=1} ^{m} \sum _{\alpha \in \mathcal{C} _{j}} N(v _{\alpha})\cdot \alpha ^{r+qs} = \sum _{j=1} ^{m} \frac{c_j ^{m-k+2+s}}{D_j}\cdot \sum _{\alpha \in \mathcal{C} _{j}} \alpha ^{r-s}.
  \]
  Since \( |r-s| \leq k-1\leq m+2 < q+1\), \((q+1)|(r-s)\) if and only if \(r =s\). Hence, \(S _{r+qs} ^{(H)} = 0\) for all \(r\neq s\), since \(\mathcal{C} _{j}\) is a coset of the \(q+1\)-th roots of unity. On the other hand, for case \(r =s\), we have
  \[
    S _{r+qr} ^{(H)} = (q+1) \cdot \sum _{j=1} ^{m} \frac{c_j ^{m-k+2+r}}{D_j} = 0,
  \]
  by \zcref{solution} and \(m-k+2+r \leq  m-2\).

  Moreover, we have

  \begin{align*}
    S_{(k-3)(q+1)}^{(H)}    & = \sum_{j=1}^{m} \frac{c_j^{m-1}}{D_j} = 1,                                                     &
    S_{(k-2)(q+1)}^{(H)}    & = \sum_{j=1}^{m} \frac{c_j^{m}}{D_j} = \Sigma,                                                    \\
    S_{(k-1)(q+1)}^{(H)}    & = \sum_{j=1}^{m} \frac{c_j^{m+1}}{D_j} = \Lambda,                                               &
    S_{(k-3)(q+1)+q}^{(H)}  & = \sum_{j=1}^{m} \frac{c_j^{m-1}}{D_j} \cdot \sum_{\alpha \in \mathcal{C}_{j}} \alpha^{q} = 0,    \\
    S_{(k-3)(q+1)+2q}^{(H)} & = \sum_{j=1}^{m} \frac{c_j^{m-1}}{D_j} \cdot \sum_{\alpha \in \mathcal{C}_{j}} \alpha^{2q} = 0, &
    S_{(k-2)(q+1)+q}^{(H)}  & = \sum_{j=1}^{m} \frac{c_j^{m}}{D_j} \cdot \sum_{\alpha \in \mathcal{C}_{j}} \alpha^{q} = 0.
  \end{align*}
  Let \(M_2 = -\text{diag}(1,\Sigma,\Lambda)\). It is easy to check that \(A_3 \overline{A_3} ^{\T} = M_2\). Hence, the second condition in \zcref{thm:SO3} holds.
  The lower bound of the minimum distance is a direct consequence of \zcref{Dist_s3}.

\end{proof}

\begin{thm}\label{HSO_s3_2}
  Let \(q \geq 5\) be a prime power, \(2 \leq  m\leq  q+1\), and \(4\leq k \leq  \frac{q+1}{2}\). Then there exists a \(q ^{2}\)-ary Hermitian self-orthogonal GRL code with parameters
  \[
    [m(q-1)+3, k, \geq m(q-1)-k+2] _{q ^{2}}.
  \]
\end{thm}
\begin{proof}
  Let \(\{\omega_1,\cdots \omega _{m}, \omega _{m+1},\cdots ,\omega _{q+1}\}\) be a representative set of \(\Fq ^{*}\) in \(\FQQ ^{*}\). We choose \(n=m(q-1)\) pairwise distinct elements in \(\FQQ ^{*}\) as follows:
  \[
    \ba : = \bigcup _{j=1} ^{m} \omega _{j} \cdot \Fq ^{*} = \{\alpha _{j,x} = \omega _{j} \cdot x: 1\leq j \leq m, x\in \Fq ^{*}\}.
  \]
  Similar to the proof of \zcref{thm2:SO2NMDS}, we can choose \(u_1,u_2,\cdots ,u _{m} \in \Fq ^{*}\) (not necessarily distinct) such that
  \[
    \mu := \sum _{j=1} ^{m} u_j \omega _{j} ^{k-3+q(k-1)} \neq 0, \text{ and } \varpi: =\sum _{j=1}^{m} u_j \omega _{j} ^{(k-2)(q+1)} \neq 0.
  \]
  For each \(\alpha _{j,x} \in \ba\), we choose \(v _{j,x} \in \FQQ ^{*}\) with \(N(v _{j,x}) = u _{j}\cdot x ^{q-2k+3}.\) Moreover, let \(a,b,c \in \FQQ ^{*}\) such that \(N(a) = -\overline{\mu}\), \(N(b) = -\varpi\), and \(N(c)=-\mu\). We set
  \[
    A_3 = \begin{pmatrix}
      0 & 0 & c \\
      0 & b & 0 \\
      a & 0 & 0
    \end{pmatrix}.
  \] Then we claim that the \(\GRL _{k}(\ba,\bv,A_3)\) is a Hermitian self-orthogonal code.

  For \(0\leq r \leq k-4\), and \(0\leq s \leq k-1\), we have
  \begin{align*}
    S _{r+qs} ^{(H)} & = \sum _{j=1} ^{m} \sum _{x \in \Fq ^{*}} N(v _{j,x}) \cdot \alpha _{j,x} ^{r+qs} = \sum _{j=1} ^{m} u_j \cdot \omega _{j} ^{r+qs} \cdot \sum _{x\in \Fq ^{*}} x ^{r+s+q-2k+3}.
  \end{align*}
  Since \(r+s+q -2k+3 \leq  q-2\), we have \(S _{r+qs} ^{(H)} = 0\) for all \(r\) and \(s\). Moreover,
  \begin{align*}
    S _{(k-3)(1+q)} ^{(H)}    & = \sum _{j=1} ^{m} u_j \cdot \omega _{j} ^{(k-3)(1+q)} \cdot \sum _{x \in\Fq ^{*}} x ^{q-3} = 0,  \\
    S _{(k-3)(1+q)+q} ^{(H)}  & = \sum _{j=1}^{m} u_j \cdot \omega _{j}^{q+(k-3)(1+q)} \cdot \sum _{x\in \Fq ^{*}} x ^{q-2} = 0,  \\
    S _{(k-3)(1+q)+2q} ^{(H)} & = \sum _{j=1} ^{m} u_j \cdot \omega _{j} ^{k-3+q(k-1)} = \mu,                                     \\
    S _{(k-2)(1+q)} ^{(H)}    & = \sum _{j=1} ^{m} u_j \cdot \omega _{j} ^{(k-2)(1+q)} = \varpi,                                  \\
    S _{(k-2)(1+q)+q} ^{(H)}  & = \sum _{j=1} ^{m} u_j \cdot \omega _{j} ^{(k-2)(1+q)+q} \cdot \sum _{x \in \Fq ^{*}} x ^{q} = 0, \\
    S _{(k-1)(1+q)} ^{(H)}    & = \sum _{j=1} ^{m} u_j \cdot \omega _{j} ^{(k-1)(1+q)} \cdot \sum _{x \in \Fq ^{*}} x ^{q+1} = 0.
  \end{align*}
  Let \(M_2 = \begin{pmatrix}
    0               & 0       & -\mu \\
    0               & -\varpi & 0    \\
    -\overline{\mu} & 0       & 0
  \end{pmatrix}\). It is easy to check that \(A_3 \overline{A_3} ^{\T} = M_2\). Hence, the second condition in \zcref{thm:SO3} holds. The lower bound of the minimum distance is a direct consequence of \zcref{Dist_s3}.
\end{proof}

\section{Quantum GRL codes}

In this section, we apply the Hermitian self-orthogonal GRL codes constructed in the previous section to build quantum error-correcting codes, and then compare their parameters with the quantum Singleton bound.

We firstly recall some basic notations and definitions. Please refer to \cite{ Ashikhmin_2001,grassl1999quantum, Ketkar_2006} for more details about quantum error-correcting codes.

Let \(\mathbb{C}\) be the complex field, \(q\) a prime power of prime \(p\), \(\mathbb{C} ^{q}\) the \(q\)-dimensional complex Hilbert space, and \(\mathcal{V} _{n} = \bigotimes _{i=1} ^{n} \mathbb{C} ^{q}\). Then \(\mathcal{V}_n\) has an orthonormal basis
\[
  \left\{ \ket{x_1 x_2 \cdots x_n} := \ket{x_1} \otimes \ket{x_2} \otimes  \cdots  \otimes \ket{x_n}\mid x_i \in \Fq, 1\leq i \leq n \right\}.
\]
Here \(\ket{x _{i}}\) is the orthonormal basis of \(\mathbb{C} ^{q}\) written in Dirac notation. A quantum error-correcting code (QECC) \(\mathcal{Q}\) is a \(K\)-dimensional subspace of \(\mathcal{V} _{n}\).

For a quantum state \(\ket{x} \in \mathbb{C} ^{q}\) with \(x \in \Fq\), and \(\forall a, b \in \Fq \), let
\[
  X(a) \ket{x} = \ket{x+a}, \quad Z(b) \ket{x} = \omega ^{\text{Tr}(bx)} \ket{x}, \quad \text{respectively},
\]
where \(\omega\) is a primitive \(p\)-th root of unity, and \(\text{Tr}\) is the trace function from \(\Fq\) to \(\mathbb{F}_p\). Moreover, for \(\mathbf{a} = (a_1,\cdots ,a_n) \in \Fq ^{n}\), we define
\[
  X(\mathbf{a}) = X(a_1) \otimes X(a_2) \otimes \cdots \otimes X(a_n), \quad Z(\mathbf{a}) = Z(a_1) \otimes Z(a_2) \otimes \cdots \otimes Z(a_n).
\]
Then the set
\[
  \mathcal{G}_n = \{ \omega ^{c} X(\mathbf{a}) Z(\mathbf{b}) : c \in \mathbb{Z} _{p}, \mathbf{a}, \mathbf{b} \in \Fq ^{n}\}
\]
is called the error group associated with \(\{ X(\ba)Z(\mathbf{b})\mid a,b \in \Fq ^{n}\}\). For any \( E = \omega ^{c}X(\mathbf{a})Z(\mathbf{b}) \in \mathcal{G} _{n}\), the quantum weight of \(E\) is defined as the number of \(i\) such that \((a_i,b_i)\neq (0,0)\).

For a QECC \(\mathcal{Q}\), we say \(\mathcal{Q}\) has minimum distance (weight) \(d(\mathcal{Q})\) if it can detect all errors in \(G _{n}\) of weight less than \(d(\mathcal{Q})\), but it cannot detect some error in \(G _{n}\) of weight \(d(\mathcal{Q})\). We say \(\mathcal{Q}\) is an \([[n,k,d]]_q\)-quantum code if \( \mathcal{Q} \subset \mathcal{V}_n \), \( \dim (\mathcal{Q}) = q ^{k} \) and \(d(\mathcal{Q}) = d\).

A well-known and widely used method to construct quantum codes is the Calderbank-Shor-Steane (CSS) construction \cite{Calderbank_1998}, which allows us to construct quantum codes from classical linear codes. We recall the CSS construction as follows.
\begin{lemma}\label{CSS}
  Let \(\C\) be a linear code over \(\FQQ\) with parameters \([n,k] _{q ^{2}}\). If \(\C\) is Hermitian self-orthogonal, then there exists a QECC \(\mathcal{Q}\) with parameters \([[n,n-2k, d(\mathcal{Q})]]\),
  where
  \[
    d(\mathcal{Q}) = \begin{cases}
      d(\C)                           & \text{if } d(\C) = d(\C ^{\perp _{H}}), \\
      \wt(\C ^{\perph} \backslash \C) & \text{otherwise},
    \end{cases}
  \]
  where \( \wt(\cdot)\) is the minimum Hamming weight of the set.

\end{lemma}

CSS construction is a powerful tool for constructing quantum codes, and many famous families are obtained in this way, such as \cite{beth1998quantum,grassl1999quantum, Steane_1999,  Xiao_2010}. Following this standard framework, we call the quantum code constructed from a Hermitian self-orthogonal GRL code via CSS construction a quantum GRL code.

\begin{Def}
  Let \(\GRL _{k}(\ba,\bv, A_s)\) be a Hermitian self-orthogonal GRL code over \(\FQQ\).   Then we call the quantum code constructed from \(\GRL _{k}(\ba,\bv, A_s)\) by \zcref{CSS} a quantum GRL code.
\end{Def}

As with classical linear codes, the commonly used method to show the optimality of a quantum code is to compare its parameters with the Singleton bound.

\begin{lemma}[Quantum Singleton bound]\cite{Ashikhmin_2001}

  Let \(\mathcal{Q}\) be an \([[n,k,d]]_q\)-quantum code. Then
  \(
  n-k \geq 2(d-1).
  \)

\end{lemma}
In \cite[Definition 2.5]{Liu_2020}, Liu, and Hu introduced quantum near-maximum-distance-separable (QNMDS) codes for describing quantum codes whose parameters are close to the quantum Singleton bound. We recall this definition as follows.

\begin{Def}\label{QNMDS}
  A quantum code \(\mathcal{Q}\) with parameters \([[n,k,d]]_{q}\) is called a quantum near-maximum-distance-separable (QNMDS) code if it satisfies
  \[
    2d \geq n-k.
  \]
\end{Def}

With these preliminaries in place, we now apply the Hermitian self-orthogonal GRL constructions to obtain quantum GRL codes and compare their parameters with the quantum Singleton bound. We have the following theorems.

\begin{thm}\label{QGRL1}
  Let \(q\) be a prime power with \(q\geq 4\), and \(m,k\) integers with \(2\leq m\leq q-2\), and \(3\leq  k\leq m+1\). Then there exists a quantum GRL code with parameters \[ [[m(q+1)+2, m(q+1)+2-2k, k]]_q. \] \end{thm}
\begin{proof}
  By \zcref{thm:SO2NMDS}, there exists a Hermitian self-orthogonal GRL code \(\C\) which is also an NMDS code with parameters \([m(q+1)+2, k, m(q+1)+2-k]_{q^2}\). By \zcref{CSS}, there exists a quantum code \(\mathcal{Q}\) with parameters \([[m(q+1)+2, m(q+1)+2-2k, d(\mathcal{Q})]]_q\).

  Moreover, since \(2k\leq  2m+2 < m(q+1)+2\), \(\C\) cannot be Hermitian self-dual. We have \(d(\mathcal{Q}) = \wt (C ^{\perph}\backslash C)\). On the other hand, \( d(\C ^{\perph}) =k\). Because \( d(C ^{\perph}) = k < m(q+1)+2- k = d(C)\), we have \(d(\mathcal{Q}) = k\). Hence, the quantum code \(\mathcal{Q}\) has parameters \[[[m(q+1)+2, m(q+1)+2-2k, k]]_q.\]
\end{proof}

\begin{thm}\label{QGRL2}
  Let \(q \geq 5\), \(2\leq  m\leq q+1\), and \(3\leq k\leq \frac{q+1}{2}\). Then there exists a quantum GRL code with parameters
  \[
    [[m(q-1)+2, m(q-1)+2-2k,k]]_q.
  \]
\end{thm}
\begin{proof}
  Similar to the proof of \zcref{QGRL1}, by \zcref{thm2:SO2NMDS} and \zcref{CSS}, there exists a quantum code \(\mathcal{Q}\) with parameters \[[[m(q-1)+2, m(q-1)+2-2k, d(\mathcal{Q})]]_q,\]
  which comes from a Hermitian self-orthogonal GRL code \( \C\) with parameters \( [m(q-1)+2, k, m(q-1)+2-k]_{q^2}\).
  Since \(d(\C^{\perp})=k < m(q-1)+2-k =d(\C)\), we have \(d(\mathcal{Q}) = k\). Hence, the quantum code \(\mathcal{Q}\) has parameters \[[[m(q-1)+2, m(q-1)+2-2k,k]]_q.\]
\end{proof}

\begin{rk}
  By \zcref{QNMDS}, the quantum codes constructed in \zcref{QGRL1} and \zcref{QGRL2} are QNMDS codes.
\end{rk}

\begin{ex}
  Let \( q = 9 \), \( \Fq = \mathbb{F}_{3}(\alpha) \) with \( \alpha^2 = -1 \), and \( \FQQ = \Fq(\omega) \) with \( \omega^2 = \alpha + 1 \). Set \( g = \alpha + \omega \), a primitive element of \( \FQQ^\times \). Let \( \xi = g^8 \), which has order 10.
  We define \( U_{10} = \langle \xi \rangle \), \( \beta_j = g^{j-1} \) for \( j = 1, \dots, 6 \), and
  \[
    (u_1, u_2, u_3, u_4, u_5, u_6) = (g^4, g^2, g^7, g^2, g^3, g^3).
  \]
  Let
  \begin{align*}
    \ba  = \bigcup_{j=1}^6 \beta_j \cdot U_{10},\qquad
    \bv  = \left( \underbrace{u_1, \dots, u_1}_{10}, \underbrace{u_2, \dots, u_2}_{10}, \dots, \underbrace{u_6, \dots, u_6}_{10} \right),
  \end{align*}
  and
  \[
    A_2 = \begin{pmatrix}
      g^5    & g^{30} \\
      g^{31} & g^6
    \end{pmatrix}.
  \]
  Then the GRL code \( \mathrm{GRL}_6(\ba, \bv, A_2) \) is a Hermitian self-orthogonal NMDS code with parameters \([62, 6, 56]_{81}\). This code gives a QECC with parameters \([[62, 50, 6]]_9\) by \(\zcref{CSS}\), which has a larger minimum distance than the previously best-known quantum code with parameters \([[62, 50, 4]]_9\) in \(\cite{BE:tables}\).
\end{ex}
\begin{proof}
  Let
  \[ (c_1,c_2,c_3,c_4,c_5,c_6) = (\beta_1 ^{10}, \beta_2 ^{10}, \beta_3 ^{10}, \beta_4 ^{10}, \beta_5 ^{10}, \beta_6 ^{10}) = (1,1+2\alpha, \alpha, 1+\alpha, 2, 2+\alpha).\]
  Then we have \( u_j ^{10} = c_j /D_j\) for \( j = 1, \dots, 6 \), where \( D_j = \prod _{\ell \neq j} (c_j - c_\ell) \). We set \( \Sigma:= \sum _{j=1} ^{6} c_j = g ^{10}\), and
  \[ \theta_1 = \beta_1 \cdot \frac{1- \xi ^{5}}{1-\xi} = g ^{26}.\] Then \( \Sigma + \theta_1:  = 1+\alpha \neq 0\), and
  \[ \rho = - \frac{\Sigma}{\Sigma+ \theta_1 ^{10}} = g ^{50} \in \Fq ^{*}.\]
  We choose \( a = g ^{5}\), \( d=g ^{6}\), which follows
  \[ a ^{10} = \rho, \qquad d ^{10} = - \frac{\Sigma ^{2}}{\Sigma + \theta_1 ^{10}} = 2\alpha = g ^{60}.\]
  Hence, we have
  \[ A_2 = \begin{pmatrix}
      a                & - \rho \cdot \overline{(\theta_1 /d)} \\
      \theta_1 \cdot a & d
    \end{pmatrix}.\] By the construction in the proof of \zcref{thm:SO2NMDS}, the GRL code \( \GRL _{6}(\ba,\bv, A_2)\) is a Hermitian self-orthogonal NMDS code with parameters \([62, 6, 56]_{81}\). By \zcref{CSS}, there exists a quantum code \(\mathcal{Q}\) with parameters \([[62, 50, 6]]_9\). We also double-check the parameters of the resulting quantum code by Magma \cite{Magma}.
\end{proof}

We next take GRL codes with \(s=3\) as ingredients for quantum constructions. Unlike the \(s=2\) case, the exact minimum distance of the dual code is difficult to determine. Nevertheless, we can derive a useful lower bound, which is sufficient for establishing the quality of the resulting quantum codes.

\begin{lemma}\cite[Theorem 4.1]{Li_2025}\label{Hmatrix_s3}
  Let \(\ba= \{\alpha_1,\cdots ,\alpha _{n}\} \subset \Fq\) be a subset of \(\Fq\), \(\bv \in (\Fq ^{*}) ^{n}\), and \(A_3\) a nonsingular \(3 \times  3\) matrix over \(\Fq\).

  Then the parity check matrix of \(\GRL _{k}(\ba,\bv,A_3)\) is
  \[
    \begin{pNiceArray}{w{c}{0.2cm}w{c}{0.6cm}|c}[margin]
      \Block{2-2}{H_{n-k}} && O \\
      && -T_3(A_3^{-1})^{\T}
    \end{pNiceArray}_{n-k \times (n+3)}
  \]
  where
  \[
    H _{n-k} = \begin{pmatrix}
      u_1                   & u_2                   & \cdots & u_n                   \\
      u_1 \alpha_1          & u_2 \alpha_2          & \cdots & u_n \alpha_n          \\
      \vdots                & \vdots                & \ddots & \vdots                \\
      u_1 \alpha_1 ^{n-k-1} & u_2 \alpha_2 ^{n-k-1} & \cdots & u_n \alpha_n ^{n-k-1}
    \end{pmatrix},
  \]
  with \(u_i = \frac{w_i}{v_i}\), and \(w_i = \prod _{j=1,j\neq i}^{n} \frac{1}{\alpha _{i}- \alpha _{j}}\) for \(1\leq i \leq n\), and
  \[
    T_3 = \begin{pmatrix}
      0 & 0              & 1                             \\
      0 & 1              & - \sigma_1(\ba)               \\
      1 & -\sigma_1(\ba) & \sigma_1 ^{2} - \sigma_2(\ba)
    \end{pmatrix}
  \]
  with \(\sigma_1(\ba) = -\sum _{i=1} ^{n} \alpha _{i}\) and \(\sigma_2(\ba) = \sum _{1\leq i < j \leq n} \alpha _{i} \alpha _{j}\).

\end{lemma}

\begin{thm}\label{QGRL3}
  Let \(q \geq  5\) be a prime power, \(2 \leq m\leq  q-3\), and \(4 \leq  k \leq m+2\). Then there exists a quantum GRL code with parameters
  \[
    [[m(q+1)+3, m(q+1)+3 - 2k , \geq k-1]] _{q}.
  \]
\end{thm}
\begin{proof}
  By \zcref{HSO_s3}, there exists a Hermitian GRL code \(\C\) with parameters \([m(q+1)+3,k] _{q ^{2}}\). By \zcref{Hmatrix_s3}, the Euclidean dual code \(C ^{\perp}\) is still a GRL code with parameters \([m(q+1)+3, m(q+1)+3-k]_{q ^{2}}\).

  Moreover, we have \(d(\C ^{\perp}) \geq k-1\) by \zcref{Dist_s3}. As the Hermitian dual code of a linear code with same parameters with its Euclidean dual code, we have \(d(\C ^{\perph})\geq k-1\). Hence, by \zcref{CSS}, there exists a quantum GRL code with parameters
  \[
    [[m(q+1)+3, m(q+1)+3 - 2k , \geq k-1]] _{q}.
  \]

\end{proof}

\begin{thm}\label{QGRL4}
  Let \(q \geq 5\) be a prime power, \(2 \leq m \leq q+1\), and \(4\leq  k\leq  \frac{q+1}{2}\). Then there exists a quantum GRL code with parameters
  \[
    [[m(q-1)+3, m(q-1)+3 - 2k, \geq k-1]] _{q}.
  \]
\end{thm}
\begin{proof}
  This proof is similar to that of the previous theorem. By \zcref{HSO_s3_2}, there exists a Hermitian self-orthogonal GRL code \(\C\) with parameters \([m(q-1)+3,k] _{q ^{2}}\), and the minimum distance of the Hermitian dual code \(\C ^{\perph}\) is at least \(k-1\). Hence, by \zcref{CSS}, there exists a quantum GRL code with parameters
  \[
    [[m(q-1)+3, m(q-1)+3 - 2k, \geq k-1]] _{q}.
  \]
\end{proof}

In table II, we list the parameters of the quantum GRL codes constructed in \zcref{QGRL1}, \zcref{QGRL2}, \zcref{QGRL3}, and \zcref{QGRL4}, which are new or improved quantum codes by comparing with the known quantum codes in \cite{BE:tables}, \cite{table2}, \cite{table3}, \cite{table4}. To more clearly compare the parameters of our quantum GRL codes with those of the known quantum codes, similar to the Singleton defect of classical linear codes, for an \( [[n,k,d]]_q \) QECC \(\mathcal{Q}\), we define \(S(\mathcal{Q})=\frac{n - k}{2}+1 - d\). This is called the quantum Singleton defect (see \cite{Pereira_2022}). Then if \( S(\mathcal{Q}) = 0\), the quantum code \(\mathcal{Q}\) is a quantum MDS code, and if \( S(\mathcal{Q}) = 1\), the quantum code \(\mathcal{Q}\) is a quantum near MDS code.

\begin{table}[!ht]
  \centering
  \label{tab:q9codes}
  \caption{New and improved q-ary QECCs from quantum GRL codes.}
  \small
  \setlength{\tabcolsep}{1.1pt}
  \renewcommand{\arraystretch}{1.20}
  \begin{tabular}{c c c c c @{\quad} c c c c c}

    \cmidrule(lr){1-5} \cmidrule(lr){6-10}
    \textbf{Our codes}   & $S(\mathcal{Q})$ & \textbf{Source} & \textbf{Known QECCs}              & \( S(\mathcal{Q})\) &
    \textbf{Our codes}   & $S(\mathcal{Q})$ & \textbf{Source} & \textbf{Known QECCs}              & $S(\mathcal{Q})$                                                                                                        \\
    \midrule
    $[[19,11,\geq 3]]_9$ & 2                & \zcref{QGRL4}   & --                                & --                  & $[[72,64,4]]_9$             & 1 & \zcref{QGRL1} & --                                        & --  \\
    $[[19,9,\geq 4]]_9$  & 2                & \zcref{QGRL4}   & --                                & --                  & $[[72,62,5]]_9$             & 1 & \zcref{QGRL1} & --                                        & --  \\
    $[[22,16,3]]_9$      & 1                & \zcref{QGRL1}   & --                                & --                  & $[[72,60,6]]_9$             & 1 & \zcref{QGRL1} & --                                        & --  \\
    $[[23,15,\geq 3]]_9$ & 2                & \zcref{QGRL3}   & --                                & --                  & $[[72,58,7]]_9$             & 1 & \zcref{QGRL1} & $[[82,58,7]]_9$~\cite{BE:tables}          & 6   \\
    $[[26,18,4]]_9$      & 1                & \zcref{QGRL2}   & $[[26,17,4]]_9$~\cite{BE:tables}  & 1.5                 & $[[72,56,8]]_9$             & 1 & \zcref{QGRL1} & $[[72,0,8]]_9$~\cite{BE:tables}           & 29  \\
    $[[26,16,5]]_9$      & 1                & \zcref{QGRL2}   & $[[27,16,5]]_9$~\cite{BE:tables}  & 1.5                 & $[[74,68,3]]_9$             & 1 & \zcref{QGRL2} & --                                        & --  \\
    $[[34,26,4]]_9$      & 1                & \zcref{QGRL2}   & --                                & --                  & $[[74,66,4]]_9$             & 1 & \zcref{QGRL2} & --                                        & --  \\
    $[[34,24,5]]_9$      & 1                & \zcref{QGRL2}   & --                                & --                  & $[[74,64,5]]_9$             & 1 & \zcref{QGRL2} & --                                        & --  \\
    $[[34,28,3]]_9$      & 1                & \zcref{QGRL2}   & $[[36,28,3]]_9$~\cite{BE:tables}  & 2                   & $[[75,65,\geq 4]]_9$        & 2 & \zcref{QGRL4} & --                                        & --  \\
    $[[42,32,5]]_9$      & 1                & \zcref{QGRL1}   & --                                & --                  & $[[75,67,\geq 3]]_9$        & 2 & \zcref{QGRL4} & $[[75,64,3]]_9$~\cite{BE:tables}          & 3.5 \\
    $[[43,35,\geq 3]]_9$ & 2                & \zcref{QGRL3}   & --                                & --                  & $[[83,75,\geq 3]]_9$        & 2 & \zcref{QGRL4} & --                                        & --  \\
    $[[43,33,\geq 4]]_9$ & 2                & \zcref{QGRL3}   & --                                & --                  & $[[83,73,\geq 4]]_9$        & 2 & \zcref{QGRL4} & --                                        & --  \\
    $[[43,31,\geq 5]]_9$ & 2                & \zcref{QGRL3}   & --                                & --                  & $[[22,10,6]]_{11}$          & 1 & \zcref{QGRL2} & $[[30,10,6]]_{11}$~\cite{table2}          & 5   \\
    $[[50,42,4]]_9$      & 1                & \zcref{QGRL2}   & --                                & --                  & $[[51,39,\geq 5]]_{11}$     & 2 & \zcref{QGRL3} & $[[57,39,5]]_{11}$~\cite{table2}          & 5   \\
    $[[50,40,5]]_9$      & 1                & \zcref{QGRL2}   & --                                & --                  & $[[62,54,4]]_{11}$          & 1 & \zcref{QGRL1} & $[[66,54,4]]_{11}$~\cite{table2}          & 3   \\
    $[[50,44,3]]_9$      & 1                & \zcref{QGRL2}   & $[[50,30,3]]_9$~\cite{BE:tables}  & 8                   & $[[26,20,3]]_{13}$          & 1 & \zcref{QGRL2} & $[[26,18,3]]_{13}$~\cite{table2}          & 2   \\
    $[[51,41,\geq 4]]_9$ & 2                & \zcref{QGRL4}   & --                                & --                  & $[[26,18,4]]_{13}$          & 1 & \zcref{QGRL2} & $[[26,18,3]]_{13}$~\cite{table2}          & 2   \\
    $[[51,43,\geq 3]]_9$ & 2                & \zcref{QGRL4}   & $[[51,39,3]]_9$~\cite{BE:tables}  & 4                   & $[[26,16,5]]_{13}$          & 1 & \zcref{QGRL2} & $[[26,14,5]]_{13}$~\cite{table2}          & 2   \\
    $[[52,42,5]]_9$      & 1                & \zcref{QGRL1}   & $[[53,42,5]]_9$~\cite{BE:tables}  & 1.5                 & $[[26,14,6]]_{13}$          & 1 & \zcref{QGRL2} & $[[26,14,5]]_{13}$~\cite{table2}          & 2   \\
    $[[52,40,6]]_9$      & 1                & \zcref{QGRL1}   & $[[52,38,6]]_9$~\cite{BE:tables}  & 2                   & $[[26,12,7]]_{13}$          & 1 & \zcref{QGRL2} & $[[26,12,6]]_{13}$~\cite{table2}          & 2   \\
    $[[58,48,5]]_9$      & 1                & \zcref{QGRL2}   & --                                & --                  & $[[98,90,4]]_{25}$          & 1 & \zcref{QGRL2} & $[[110,90,4]]_{25}$~\cite{table3}         & 7   \\
    $[[58,50,4]]_9$      & 1                & \zcref{QGRL2}   & $[[62,50, 4]]_9$~\cite{BE:tables} & 3                   & $[[243,229,\geq 6]]_{81}$   & 2 & \zcref{QGRL4} & $[[243,217,\geq 6]]_{81}$~\cite{table4}   & 8   \\
    $[[62,52,5]]_9$      & 1                & \zcref{QGRL1}   & --                                & --                  & $[[243,225,\geq 8]]_{81}$   & 2 & \zcref{QGRL4} & $[[243,207,\geq 8]]_{81}$~\cite{table4}   & 11  \\
    $[[62,48,7]]_9$      & 1                & \zcref{QGRL1}   & --                                & --                  & $[[243,221,\geq 10]]_{81}$  & 2 & \zcref{QGRL4} & $[[243,197,\geq 10]]_{81}$~\cite{table4}  & 14  \\
    $[[62,56,3]]_9$      & 1                & \zcref{QGRL1}   & $[[65,56,3]]_9$~\cite{BE:tables}  & 2.5                 & $[[243,217,\geq 12]]_{81}$  & 2 & \zcref{QGRL4} & $[[243,217,\geq 6]]_{81}$~\cite{table4}   & 8   \\
    $[[62,54,4]]_9$      & 1                & \zcref{QGRL1}   & $[[62,50, 4]]_9$~\cite{BE:tables} & 3                   & $[[243,207,\geq 17]]_{81}$  & 2 & \zcref{QGRL4} & $[[243,207,\geq 8]]_{81}$~\cite{table4}   & 11  \\
    $[[62,50,6]]_9$      & 1                & \zcref{QGRL1}   & $[[62,50, 4]]_9$~\cite{BE:tables} & 3                   & $[[243,197,\geq 22]]_{81}$  & 2 & \zcref{QGRL4} & $[[243,197,\geq 10]]_{81}$~\cite{table4}  & 14  \\
    $[[63,53,\geq 4]]_9$ & 2                & \zcref{QGRL3}   & --                                & --                  & $[[507,489,\geq 8]]_{169}$  & 2 & \zcref{QGRL4} & $[[507,477,\geq 8]]_{169}$~\cite{table4}  & 8   \\
    $[[63,49,\geq 6]]_9$ & 2                & \zcref{QGRL3}   & --                                & --                  & $[[507,485,\geq 10]]_{169}$ & 2 & \zcref{QGRL4} & $[[507,469,\geq 10]]_{169}$~\cite{table4} & 10  \\
    $[[63,47,\geq 7]]_9$ & 2                & \zcref{QGRL3}   & --                                & --                  & $[[507,477,\geq 14]]_{169}$ & 2 & \zcref{QGRL4} & $[[507,477,\geq 8]]_{169}$~\cite{table4}  & 8   \\
    $[[63,51,\geq 5]]_9$ & 2                & \zcref{QGRL3}   & $[[71,51,5]]_9$~\cite{BE:tables}  & 6                   & $[[507,469,\geq 18]]_{169}$ & 2 & \zcref{QGRL4} & $[[507,469,\geq 10]]_{169}$~\cite{table4} & 10  \\
    $[[66,56,5]]_9$      & 1                & \zcref{QGRL2}   & --                                & --                  & $[[867,853,\geq 6]]_{289}$  & 2 & \zcref{QGRL4} & $[[867,845,\geq 6]]_{289}$~\cite{table4}  & 6   \\
    $[[66,60,3]]_9$      & 1                & \zcref{QGRL2}   & $[[66,58,3]]_9$~\cite{BE:tables}  & 2                   & $[[867,849,\geq 8]]_{289}$  & 2 & \zcref{QGRL4} & $[[867,837,\geq 8]]_{289}$~\cite{table4}  & 8   \\
    $[[66,58,4]]_9$      & 1                & \zcref{QGRL2}   & $[[66,58,3]]_9$~\cite{BE:tables}  & 2                   & $[[867,845,\geq 10]]_{289}$ & 2 & \zcref{QGRL4} & $[[867,845,\geq 6]]_{289}$~\cite{table4}  & 6   \\
    $[[67,59,\geq 3]]_9$ & 2                & \zcref{QGRL4}   & $[[69,59,3]]_9$~\cite{BE:tables}  & 3                   & $[[867,837,\geq 14]]_{289}$ & 2 & \zcref{QGRL4} & $[[867,837,\geq 8]]_{289}$~\cite{table4}  & 8   \\
    $[[67,57,\geq 4]]_9$ & 2                & \zcref{QGRL4}   & $[[71,57,4]]_9$~\cite{BE:tables}  & 4                   & $[[867,829,\geq 18]]_{289}$ & 2 & \zcref{QGRL4} & $[[867,829,\geq 10]]_{289}$~\cite{table4} & 10  \\
    \bottomrule
  \end{tabular}
\end{table}

\section{CONCLUDING REMARKS}
In this paper, we established necessary and sufficient NMDS conditions for GRL codes with \(s=2\) and \(s=3\). In particular, the \(s=2\) criterion unifies known NMDS results for Roth-Lempel codes. Further, we constructed four new classes of Hermitian self-orthogonal codes from GRL codes with \(s=2\) and \(s=3\). Building on characterization of NMDS property, the resulting Hermitian self-orthogonal GRL codes with \( s=2\) are always NMDS codes and provide more flexible lengths and dimensions than known Hermitian self-orthogonal NMDS constructions. Finally, by applying CSS construction, we obtained quantum GRL codes and showed that those derived from the \(s=2\) families are QNMDS codes. Determining whether the Hermitian self-orthogonal GRL families with \( s=3\) can yield quantum NMDS codes remains an open problem for future work.

\bibliographystyle{IEEEtranS}
\bibliography{ref.bib}

\end{document}